\documentclass[11pt,a4paper]{article}
\usepackage[left=32mm,right=32mm,top=30mm,bottom=32mm]{geometry}
\usepackage[english]{babel}              
\usepackage{amsmath}                    
\usepackage{amssymb}  
\usepackage{amsthm}                  
\usepackage{xcolor}
\usepackage{graphicx}
\usepackage[font=small,labelfont=bf]{caption}
\usepackage{subcaption}
\usepackage{authblk}

\newtheorem{lemma}{Lemma}
\newtheorem{corollary}[lemma]{Corollary}

\newtheorem{theorem}[lemma]{Theorem}

\theoremstyle{definition}
\newtheorem{definition}[lemma]{Definition}

\newtheorem{remark}[lemma]{Remark}

\newcommand{\sys}{\operatorname{sys}}
\newcommand{\injrad}{\operatorname{injrad}}
\newcommand{\arccosh}{\operatorname{arccosh}}
\newcommand{\arcsinh}{\operatorname{arcsinh}}
\newcommand{\arctanh}{\operatorname{arctanh}}

\newcommand{\diam}{\operatorname{diam}}
\newcommand{\mindist}{\operatorname{mindist}}

\newcommand{\dist}{\operatorname{dist}}

\newcommand{\mcM}{\mathcal{M}}

\newcommand{\mcT}{T}

\newcommand{\mcP}{\mathcal{P}}
\newcommand{\mcD}{\mathcal{D}}
\newcommand{\teich}{\mathcal{T}}

\newcommand{\mM}{X}
\newcommand{\mS}{\Gamma}
\newcommand{\mN}{\mathbb{N}}
\newcommand{\mR}{\mathbb{R}}
\newcommand{\mD}{\mathbb{D}}

\newcommand{\mMthick}{\mM^{\varepsilon}_{\operatorname{thick}}}
\newcommand{\mMthin}{\mM^{\varepsilon}_{\operatorname{thin}}}

\definecolor{dgreen}{RGB}{0,150,0}
\graphicspath{{./figs/}}

\begin{document}
	
\title{Minimal Delaunay triangulations of hyperbolic surfaces}
\author[1]{Matthijs Ebbens}
\author[2]{Hugo Parlier}
\author[1]{Gert Vegter}	
\affil[1]{Bernoulli Institute for Mathematics, Computer Science and Artificial Intelligence, University of Groningen, Netherlands}
\affil[2]{Department of Mathematics, University of Luxembourg, Luxembourg}

\date{\vspace{-1cm}}

\maketitle

\begin{abstract}
Motivated by recent work on Delaunay triangulations of hyperbolic surfaces, we consider the minimal number of vertices of such triangulations. First, we will show that every hyperbolic surface of genus $g$ has a simplicial Delaunay triangulation with $O(g)$ vertices, where edges are given by distance paths. Then, we will construct a class of hyperbolic surfaces for which the order of this bound is optimal. Finally, to give a general lower bound, we will show that the $\Omega(\sqrt{g})$ lower bound for the number of vertices of a simplicial triangulation of a topological surface of genus $g$ is tight for hyperbolic surfaces as well. 
\end{abstract}

\medskip

\noindent\textbf{2012 ACM Subject Classification.} Computational Geometry, Graphs and Surfaces.

\smallskip


\smallskip

\noindent\textbf{Funding.} The second author was partially supported by ANR/FNR project SoS, INTER/ANR/16/11554412/SoS, ANR-17-CE40-0033.

\section{Introduction}\label{sec:introduction}

The classical topic of Delaunay triangulations has recently been studied in the context of hyperbolic surfaces. Bowyer's incremental algorithm for computing simplicial Delaunay triangulations in the Euclidean plane \cite{bowyer1981} has been generalized to orientable hyperbolic surfaces and implemented for some specific cases \cite{bogdanov2016,iordanov2017}. Moreover, it has been shown that the flip graph of geometric (but not necessarily simplicial) Delaunay triangulations on a hyperbolic surface is connected \cite{despre2019}. 

In this work, we consider the minimal number of vertices of a simplicial Delaunay triangulation of a closed hyperbolic surface of genus $g$. Motivated by the interest in embeddings where edges are shortest paths between their endpoints \cite{fary1948,hubard2017}, which have applications in for example the field of graph drawing \cite{tamassia2013}, we restrict ourselves to \emph{distance} Delaunay triangulations, where edges are distance paths. 

Our main result is the upper bound on the number of vertices with sharp order of growth:
\begin{theorem}\label{thm:intro}
An orientable closed hyperbolic surface of genus $g\geq 2$ has a distance Delaunay triangulation with at most $O(g)$ vertices. Furthermore, there exists a family of surfaces, $X_g$, $g\geq 2$, such that the number of vertices of any distance Delaunay triangulation grows like $\Omega(g)$. 
\end{theorem}
The above result is a compilation of Theorems~\ref{thm:upperboundDT} and~\ref{thm:lowerbound} where explicit upper and lower bounds are given. 

Another reason to study triangulations whose edges are distance paths, comes from the study of moduli spaces $\mcM_g$, which we can think of as a space of all hyperbolic surfaces of genus $g\geq 2$ up to isometry. These spaces admit natural coordinates associated to pants decompositions (the so-called Fenchel-Nielsen coordinates, see Section \ref{sec:preliminaries} for details). It is a classical theorem of Bers \cite{Bers} that any surface admits a short pants decomposition, meaning each of its simple closed geodesics is bounded by a function that only depends on the topology of the surface (but not its geometry). As these curves provide a local description of the surface, one might hope that they are also geodesically convex, meaning that the shortest distance path between any two points of a given curve is contained in the curve. It is perhaps surprising that {\it most} surfaces admit no such pants decompositions. Indeed it is known that {\it any } pants decomposition of a random surface (chosen with respect to a natural probability measure on $\mcM_g$) have at least one curve of length on the order of $g^{\frac{1}{6}-\varepsilon}$ as $g$ grows (for any fixed $\varepsilon>0$) \cite{GPY}. And it is a theorem of Mirzakhani that these same random surfaces are also of {\it diameter} on the order of $\log(g)$ \cite{Mirzakhani}. Hence the longest curve of any pants decomposition of a random surface is not convex. 

The lengths of edges in a given triangulation are another parameter set for $\mcM_g$. By the theorem above, such a parameter set can be chosen with a reasonable number of vertices such that the edges are all convex. Using the moduli space point of view, one has a function $\omega:\mcM_g\to \mN$ which associates to a surface the minimal number of vertices of any of its distance Delaunay triangulations. The above result implies that 
$$
\limsup_{g\to \infty} \max_{X\in \mcM_g}\dfrac{\omega(X)}{g}
$$
is finite and strictly positive, but for instance we do not know whether the actual limit exists.

The examples we exhibit are geometrically quite simple, as they are made by gluing hyperbolic pants, with bounded cuff lengths, in something that resembles a line as the genus grows. One might wonder whether all surfaces have this property, but we show this is not the case by exploring the quantity $\min_{X\in \mcM_g}\omega(X)$. This quantity has a precise lower bound on the order of $\Theta(\sqrt{g})$ because we ask that our triangulations be simplicial \cite{jungerman1980}. We show how to use the celebrated Ringel-Youngs construction \cite{ringel1968} to construct a family of hyperbolic surfaces that attain this bound for infinitely many genus (Theorem~\ref{thm:generallowerbound}), showing that one cannot hope for better than the simplicial lower bound in general. 

Although our results provide a good understanding on the extremal values of $\omega$, there are still plenty of unexplored questions. For example, what is the behavior of $\omega$ for a random surface (using Mirzakhani's notion of randomness \cite{Mirzakhani} alluded to above)?
 
This paper is structured as follows. In Section~\ref{sec:preliminaries}, we will introduce our notation and give some preliminaries on hyperbolic surface theory and triangulations. In Section~\ref{sec:upperbound}, we will prove our linear upper bound for the number of vertices of a minimal distance Delaunay triangulation. In Section~\ref{sec:upperboundtight}, we will construct classes of hyperbolic surfaces attaining the order of this linear upper bound. Finally, in Section~\ref{sec:lowerbound}, we will construct a family of hyperbolic surfaces attaining the general $\Theta(\sqrt{g})$ lower bound. Proofs of two technical lemmas appear in the appendices.

\paragraph{Acknowledgements.} The authors warmly thank Monique Teillaud and Vincent Despr\'e for fruitful discussions.

\section{Preliminaries}\label{sec:preliminaries}

We will start by recalling some hyperbolic geometry. There are several models for the hyperbolic plane \cite{beardon2012}. In the Poincar\'e disk model, the hyperbolic plane is represented by the unit disk $\mD$ in the complex plane equipped with a specific Riemannian metric of constant Gaussian curvature $-1$. With respect to this metric, hyperbolic lines, i.e., geodesics are given by diameters of $\mD$ or circle segments intersecting $\partial\mD$ orthogonally. A hyperbolic circle is a Euclidean circle contained in $\mD$. However, in general the centre and radius of a hyperbolic circle are different from the Euclidean centre and radius. 

We refer to \cite{buser2010} for all of the following facts. A \emph{hyperbolic surface} is a 2-dimensional Riemannian manifold that is locally isometric to an open subset of the hyperbolic plane \cite{stillwell2012}, thus of constant curvature $-1$. Our surfaces are assumed throughout to be {\it closed} and {\it orientable}, and because they are hyperbolic, via Gauss-Bonnet, their genus $g$ satisfies $g \geq 2$ and their area is $4\pi(g-1)$. Note that we will frequently be interested in subsurfaces of a closed surface which we think of as compact surfaces with boundary consisting of a collection of simple closed geodesics. The signature of such a subsurface is $(g',k)$ where $g'$ is its genus and $k$ is the number of boundary geodesics.

Via the uniformization theorem, any hyperbolic surface $\mM$ can be written as a quotient space $\mM=\mD/\Gamma$ of the hyperbolic plane under the action of a Fuchsian group $\Gamma$ (a discrete subgroup of the group of orientation-preserving isometries of $\mD$). The hyperbolic plane $\mD$ is the \emph{universal cover} of $\mM$ and is equipped with a projection $\pi:\mD\rightarrow\mD/\Gamma$.  

In the free homotopy class of any non-contractible closed curve lies a unique closed geodesic. If the curve is simple, then the corresponding geodesic is simple, and hence it is a straightforward topological exercise to decompose a hyperbolic surface into $2g-2$ pairs of pants by cutting along $3g-3$ disjoint simple closed geodesics. A pair of pants is a surface homeomorphic to a three times punctured sphere but we generally think of its closure, and thus of a hyperbolic pair of pants as being a surface of genus $0$ with three simple closed geodesics as boundary. Its closure is thus a surface of signature $(0,3)$. See Figure~\ref{fig:popdecomposition}).

It is a short but useful exercise in hyperbolic trigonometry to show that a hyperbolic pair of pants is determined by its three boundary lengths. This is done by cutting the pair of pants along the three geodesic paths, orthogonal to the boundary, which realize the distance between the different boundary geodesics and then arguing on the resulting right angled hexagons. Hence, the lengths of the $3g-3$ geodesics determine the geometry of each of the $2g-2$ pairs of pants, but to determine $\mM$, one needs to add twist parameters that control how the pants are pasted together. How one computes the twist coordinate is at least partially a matter of taste, and although we will not make much use of it, for completeness we follow \cite{buser2010}, where the twist is the signed distance between the points lying on the basepoints of the orthogeodesics mentioned above. 

The length and twist parameters determine $\mM$ and are called Fenchel-Nielsen coordinates. These parameters can be chosen freely in the set $(\mR^{>0})^{3g-3} \times \mR^{3g-3}$. What they determine is more than just an isometry class of surface: they determine a {\it marked} hyperbolic surface, homeomorphic to a base topological surface $\Sigma$. As the lengths and twists change, the marked surface changes, and the Fenchel-Nielsen coordinates provide a parameter set for the space of marked hyperbolic surfaces of genus $g$, called Teichm\"uller space ${\mathcal T}_g$. The underlying {\it moduli space} $\mathcal{M}_g$ can be thought of as the space of hyperbolic surfaces up to isometry, obtained from ${\mathcal T}_g$ by ``forgetting'' the marking.

Throughout the paper, lengths of closed geodesics will play an important role. As mentioned above, in the free homotopy class of a non-contractible closed curve lies a unique geodesic representative, and as the metric changes, the free homotopy class is well defined, but the length of the geodesic changes. Generally we will be dealing with a fixed surface $X\in {\mathcal T}_g$, and the length of a geodesic $\gamma$ will be denoted by $\ell(\gamma)$. Nonetheless, it is sometimes useful to think of the length of the corresponding homotopy class as a function over ${\mathcal T}_g$ which associates to $X$ the length of the geodesic corresponding to $\gamma$.
\begin{figure}[h]
	\centering
	\includegraphics[width=0.9\textwidth]{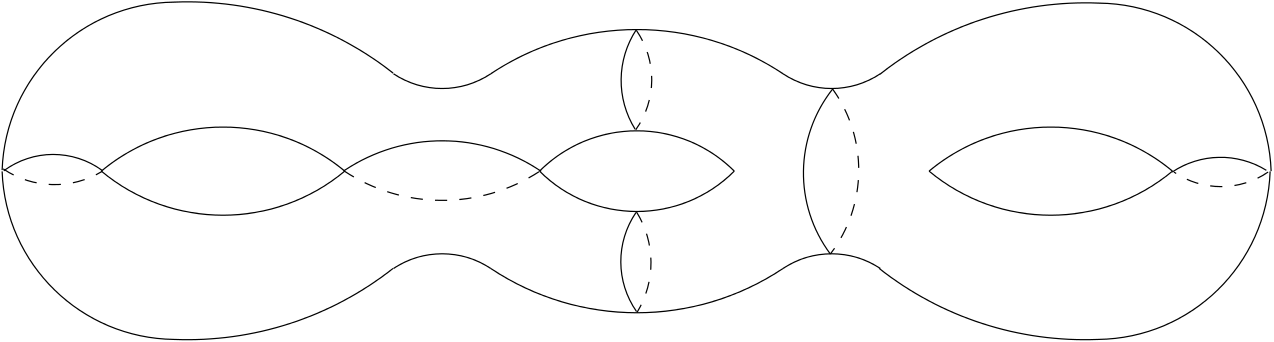}
	\caption{Decomposition of a genus 3 surface into 4 pair of pants using 6 disjoint simple closed geodesics.}
	\label{fig:popdecomposition}
\end{figure}
To a pair of pants decomposition, we can associate a 3-regular graph. In this 3-regular graph, each pair of pants is represented by a vertex and two vertices share an edge if the corresponding pairs of pants share a boundary geodesic. For example, Figure~\ref{fig:cubicgraphgenus3} shows the 3-regular graph corresponding to the pair of pants decomposition in Figure~\ref{fig:popdecomposition}. As our parametrization of $\mathcal{T}_g$ depends on a choice of pants decomposition, one can think of the Fenchel-Nielsen coordinates associating a length and a twist to each edge. 

\begin{figure}[h]
	\centering
	\includegraphics[width=0.7\textwidth]{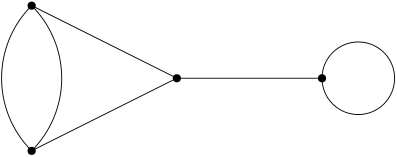}
	\caption{3-regular graph corresponding to the pair of pants decomposition shown in Figure~\ref{fig:popdecomposition}.}
	\label{fig:cubicgraphgenus3}
\end{figure}

Around a simple closed geodesic $\gamma$, the local geometry of a surface is given by its so-called collar. Roughly speaking, for small enough $r$, the set 
$$ C_\gamma(r)=\{x\in\mM\;|\;d(x,\gamma)\leq r\}$$
is an embedded cylinder. A bound on how large one can take the $r$ to be while retaining the cylinder topology is given by the Collar Lemma: 
\begin{lemma}[{\cite[Theorem 4.1.1]{buser2010}}]\label{lem:collarlemma1}
Let $\gamma$ by a simple closed geodesic on a closed hyperbolic surface $\mM$. The collar $C_\gamma(w(\gamma))$ of width $w(\gamma)$ given by
\begin{equation}\label{eq:collarwidth}
w(\gamma)=\arcsinh\bigg(  \dfrac{1}{\sinh(\tfrac{1}{2}\ell(\gamma))} \bigg)
\end{equation}
is an embedded hyperbolic cylinder isometric to $[-w(\gamma),w(\gamma)]\times\mS^1$ with the Riemannian metric $ds^2=d\rho^2+\ell^2(\gamma)\cosh^2(\rho)dt^2$ at $(\rho, t)$. Furthermore, if two simple closed geodesics $\gamma$ and $\gamma'$ are disjoint, then the collars $C_\gamma(w(\gamma))$ and $C_{\gamma'}(w(\gamma'))$ are disjoint as well.
\end{lemma}

This paper is about \emph{distance Delaunay triangulations} on closed hyperbolic surfaces. 

\begin{definition}
	A \emph{distance Delaunay triangulation} is a triangulation satisfying the following three properties:
	\begin{enumerate}
		\item it is a simplicial complex,
		\item it is a Delaunay triangulation,
		\item its edges are distance paths.
	\end{enumerate}
	The set of all distance Delaunay triangulations of a closed hyperbolic surface $\mM$ is denoted by $\mcD(\mM)$. 
\end{definition}

We will describe each of the three properties of distance Delaunay triangulations in more detail below.

\paragraph{Simplicial complexes.}
We will use the standard definition of a simplicial complex. In our case, an embedding of a graph into a surface is a simplicial complex if and only if it does not contain any $1$- or $2$-cycles. In particular, a geodesic triangulation of a point set in the Euclidean or hyperbolic planes is always a simplicial complex. This is because there are no geodesic monogons or bigons. 

\paragraph{Delaunay triangulations.}
We recall that given a set of vertices in the Euclidean plane a triangle is called a \emph{Delaunay triangle} if its circumscribed disk does not contain any vertex in its interior. A triangulation of a set of vertices in the Euclidean plane is a \emph{Delaunay triangulation} if all triangles are Delaunay triangles.

We can define Delaunay triangulations in the hyperbolic plane in the same way as in the Euclidean plane, where we use the fact that hyperbolic circles are Euclidean circles. Delaunay triangulations on hyperbolic surfaces can be defined by lifting vertices on a hyperbolic surface $\mM$ to the universal cover $\mD$ \cite{bogdanov2013b,despre2019}. More specifically, let $\mcP$ be a set of vertices on $\mM$ and let $\pi:\mD\rightarrow \mD/\Gamma$ be the projection of the hyperbolic plane $\mD$ to the hyperbolic surface $\mM=\mD/\Gamma$. A triangle $(v_1,v_2,v_3)$ with $v_i\in\mcP$ is called a \emph{Delaunay triangle} if there exist pre-images $v_i'\in\pi^{-1}(\{v_i\})$ such that the circumscribed disk of the triangle $(v_1',v_2',v_3')$ in the hyperbolic plane does not contain any point of $\pi^{-1}(\mcP)$ in its interior. A triangulation of $\mcP$ on $\mM$ is a \emph{Delaunay triangulation} if all triangles are Delaunay triangles. 

A Delaunay triangulation of a point set on a hyperbolic surface $\mM$ is related to a Delaunay triangulation in $\mD$ as follows \cite{bogdanov2013b}. Given a point set $\mcP$ on $\mM$, we consider a Delaunay triangulation $\mcT'$ of the infinite point set $\pi^{-1}(\mcP)$. Then, we let $\mcT=\pi(\mcT')$. By definition, $\mcT$ is a Delaunay triangulation. Moreover, because every triangulation in $\mD$ is a simplicial complex, $\mcT'$ is a simplicial complex. However, $\mcT$ is not necessarily a simplicial complex, because projecting $\mcT'$ to $\mM$ might introduce 1- or 2-cycles. We will use the correspondence between Delaunay triangulations in $\mD$ and in $\mM$ in Definition~\ref{def:standardtriangulation} and the proof of Theorem~\ref{thm:upperboundDT} and show explicitly that in these cases the result after projecting to $\mM$ is simplicial.

To make sure that $T=\pi(T')$ is a well-defined triangulation, we will assume without loss of generality that $\mcT'$ is \emph{$\Gamma$-invariant}, i.e., the image of any Delaunay triangle in $T'$ under an element of $\Gamma$ is a Delaunay triangle. Otherwise, it is possible that in so-called \emph{degenerate cases} $T$ contains edges that intersect in a point that is not a vertex \cite{bogdanov2016}. Namely, suppose that $\mcT'$ contains a polygon $P=\{p_1,p_2,\ldots,p_k\}$ consisting of $k\geq 4$ concircular vertices and let $T_P$ be the Delaunay triangulation of $P$ in $T'$. Because the Delaunay triangulation of a set of at least four concircular vertices is not uniquely defined, assume that there exists $A\in\Gamma$ such that the Delaunay triangulation $T_{A(P)}$ of $A(P)$ in $T'$ is not equal to $A(T_P)$. Because $\pi(P)=\pi(A(P))$, there exists an edge of $\pi(T_{A(P)})$ and an edge of $\pi(A(T_P))$ that intersect in a point that is not a vertex.

\paragraph{Distance paths.}
Suppose we are given an edge $(u,v)$ in a triangulation of a hyperbolic surface $\mM$. Because $(u,v)$ is embedded in $\mM$, there exists a geodesic $\gamma:[0,1]\rightarrow\mM$ with $\gamma(0)=u$ and $\gamma(1)=v$. We say that $(u,v)$ is a \emph{distance path} if $\ell(\gamma)=d(u,v)$, where $d(u,v)$ is the infimum of the lengths of all curves joining $u$ to $v$.

\section{Linear upper bound for the number of vertices of a minimal distance Delaunay triangulation}\label{sec:upperbound}

As our first result, we prove that for every hyperbolic surface there exists a distance Delaunay triangulation whose cardinality grows linearly as a function of the genus. Note that the constant 151 is certainly not optimal.

\begin{theorem}\label{thm:upperboundDT}
	For every closed hyperbolic surface $\mM$ of genus $g$ there exists $\mcT\in\mcD(\mM)$ with at most $151g$ vertices.
\end{theorem} 

The idea of the proof is the following. Given a hyperbolic surface $\mM$, we construct a vertex set $\mcP$ on $\mM$ consisting of at most $151g$ vertices such that the projection $\mcT$ of a Delaunay triangulation of $\pi^{-1}(\mcP)$ in $\mD$ to $\mM$ is a distance Delaunay triangulation of $\mM$. 

It is known that $\mcT$ is a simplicial complex if $\mcP$ is sufficiently dense and well-distributed \cite{bogdanov2013b}. More precisely, there are no 1- or 2-cycles in $\mcT$ if the diameter of the largest disk in $\mD$ not containing any points of $\pi^{-1}(\mcP)$ is less than $\tfrac{1}{2}\sys(\mM)$, where $\sys(\mM)$ is the systole of $\mM$, i.e. the length of the shortest homotopically non-trivial closed curve. However, the systole of a hyperbolic surface can be arbitrarily close to zero, which means that we would need an arbitrarily dense set $\mcP$ to satisfy this condition. 

Instead, for a constant $\varepsilon>0$ we subdivide $\mM$ into its \emph{$\varepsilon$-thick} part
$$ \mMthick=\{x\in\mM\;|\;\injrad(x)>\varepsilon\}$$
and its \emph{$\varepsilon$-thin} part $\mMthin=\mM\setminus\mMthick$, where $\injrad(x)$ is the injectivity radius at $x$, i.e., the radius of the largest embedded open disk centered at $x$. Note that the minimum of $\injrad(x)$ over all $x\in\mM$ is given by $\tfrac{1}{2}\sys(\mM)$. We see in Section~\ref{subsec:thinhyperboliccylinders} that $\mMthin$ is a collection of hyperbolic cylinders for sufficiently small $\varepsilon$ (see Figure~\ref{fig:thickthindecomp}). In these hyperbolic cylinders we want to construct a set of vertices of which the cardinality does not depend on $\sys(\mM)$. To do this, we put three vertices on the ``waist'' and each of the two boundary components of the cylinders that are ``long and narrow''. In the cylinders that are not ``long and narrow'' it suffices to place three vertices on its waist only. The notions of ``waist'' and ``long and narrow'' will be specified in Section~\ref{subsec:thinhyperboliccylinders}. Because $\injrad(x)>\varepsilon$ for all $x\in\mMthick$, we can construct a sufficiently dense and well-distributed point set in $\mMthick$ whose cardinality does not depend on $\sys(\mM)$ but only on $\varepsilon$. In Section~\ref{subsec:thickhyperbolicsurfaces} we will describe how we combine the vertices placed in the hyperbolic cylinders with the dense and well-distributed set of vertices in $\mMthick$. Finally, the proof of Theorem~\ref{thm:upperboundDT} is given in Section~\ref{subsec:prooftheorem1}.

\begin{figure}[h]
	\centering
	\includegraphics[width=0.9\textwidth]{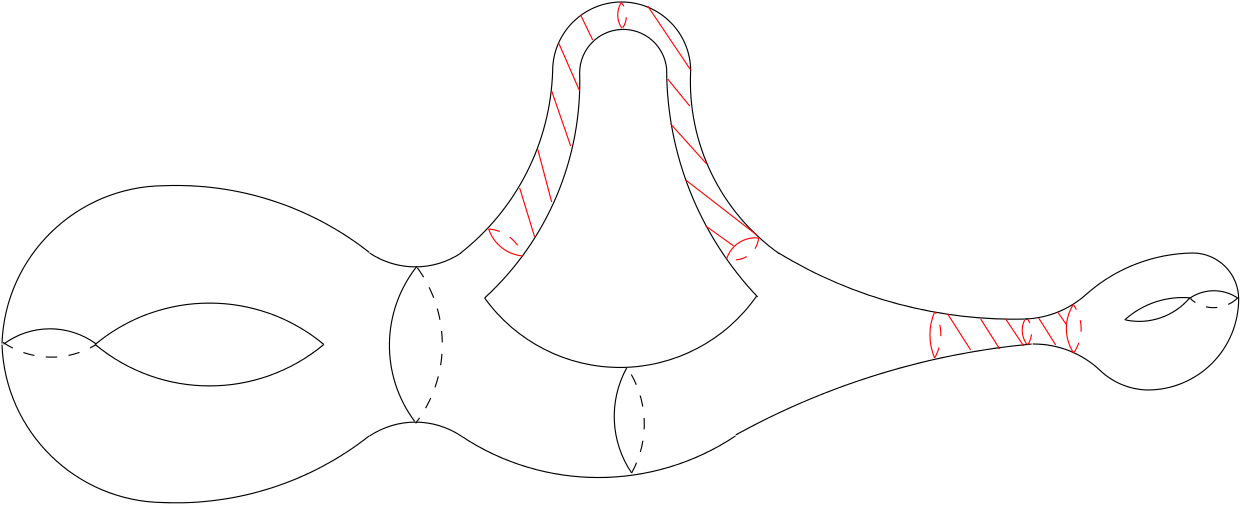}
	\caption{Decomposition of a hyperbolic surface into a thick part consisting of two connected components and two narrow hyperbolic cylinders (in red).}
	\label{fig:thickthindecomp}
\end{figure}

\subsection{Distance Delaunay triangulations of hyperbolic cylinders}\label{subsec:thinhyperboliccylinders}

We now describe our construction of a set of vertices for the $\varepsilon$-thin part $\mMthin$. The following lemma describes $\mMthin$ in more detail.

\begin{lemma}[{\cite[Theorem 4.1.6]{buser2010}}]\label{lem:collarlemma2}
	If $\varepsilon<\arcsinh(1)$ then $\mMthin$ is a collection of at most $3g-3$ pairwise disjoint hyperbolic cylinders.
\end{lemma}

The following description of the geometry of the hyperbolic cylinders in $\mMthin$ is based primarily on a similar description in the context of colourings of hyperbolic surfaces \cite{parlier2016}. Each hyperbolic cylinder $C$ in $\mMthin$ consists of points with injectivity radius at most $\varepsilon$ and the boundary curves $\gamma^+$ and $\gamma^-$ consist of all points with injectivity radius equal to $\varepsilon$. Every point on the boundary curves is the base point of an embedded geodesic loop of length $2\varepsilon$ (Figure~\ref{fig:cylinderwidth}), which is completely contained in the hyperbolic cylinder. All points on the boundary curves have the same distance $K_C$ to $\gamma$, which only depends on $\varepsilon$ and the length $\ell(\gamma)$ of $\gamma$. To see this, fix a point $p$ on $\gamma^+$ and consider a distance path $\xi$ from $p$ to $\gamma$ (Figure~\ref{fig:cylinderwidth}). Cutting along $\gamma,\xi$ and the loop of length $2\varepsilon$ with base point $p$ yields a hyperbolic quadrilateral. The common orthogonal of $\gamma$ and the geodesic loop subdivides this quadrilateral into two congruent quadrilaterals, each with three right angles. Applying a standard result from hyperbolic trigonometry yields \cite[Formula Glossary 2.3.1(v)]{buser2010}
$$\sinh(\varepsilon)=\sinh(\tfrac{1}{2}\ell(\gamma))\cosh(\ell(\xi)).$$
Because $K_C=\ell(\xi)$, it follows that
\begin{equation}\label{eq:expressionKC}
K_C=\arccosh\bigg(\dfrac{\sinh(\varepsilon)}{\sinh(\tfrac{1}{2}\ell(\gamma))} \bigg).
\end{equation}

\begin{figure}[h]
	\centering
	\includegraphics[width=0.65\textwidth]{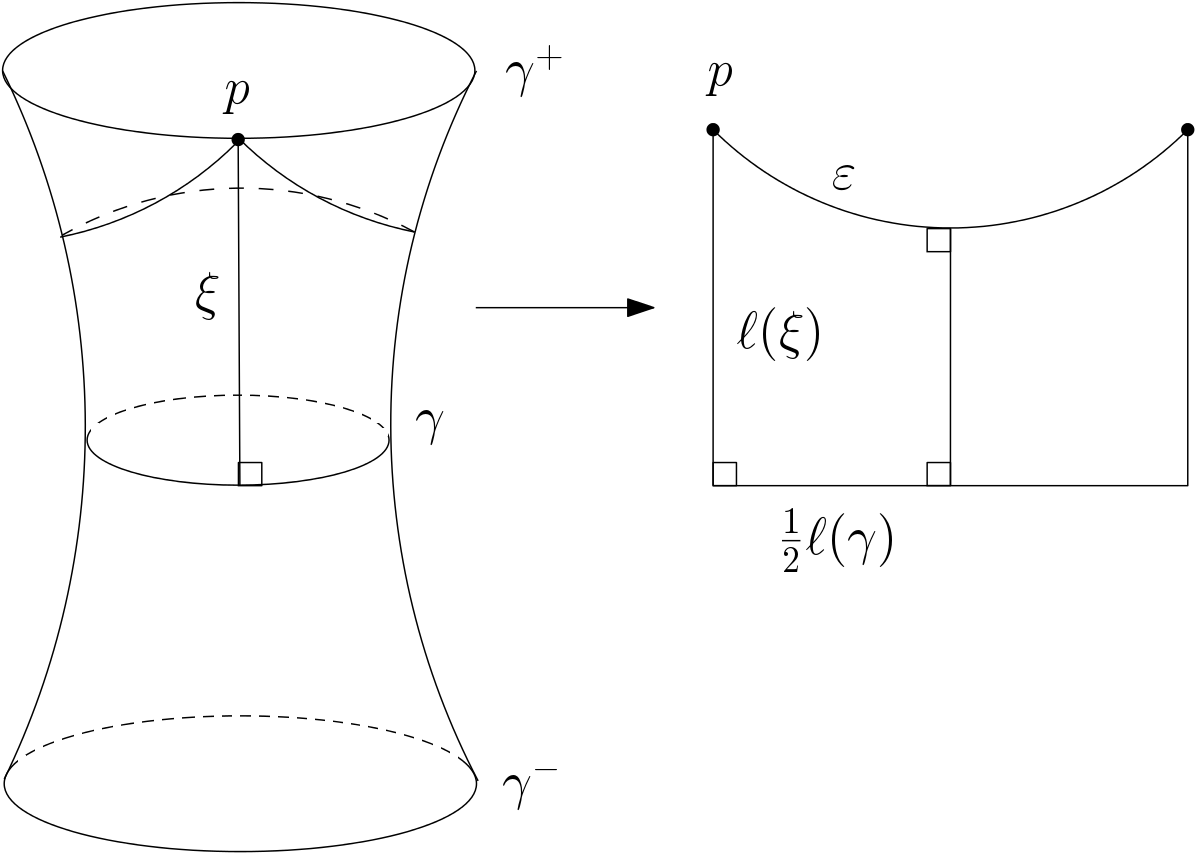}
	\caption{Computing $K_C$.}
	\label{fig:cylinderwidth}
\end{figure}

We see that $\gamma^+$ consists of points that are equidistant to $\gamma$. By symmetry, the distance between a point on $\gamma^-$ and $\gamma$ is equal to $K_C$ as well. Moreover, $\gamma^+$ and $\gamma^-$ are smooth. 

Recall the notion of a collar from Section~\ref{sec:preliminaries}. In particular, each hyperbolic cylinder $C$ in $\mMthin$ is a collar of width $K_C$. Comparing equation \eqref{eq:expressionKC} for $K_C$ with equation \eqref{eq:collarwidth} in the statement of the Collar Lemma, we see that $w(\gamma)>K_C$, because $\sinh\varepsilon<1$. This inequality will be used in the proof of Lemma~\ref{lem:bandaroundcollar} to give a lower bound for the distance between distinct hyperbolic cylinders in $\mMthin$.

We distinguish between two kinds of hyperbolic cylinders in $\mMthin$, namely $\varepsilon'$-thin cylinders and $\varepsilon'$-thick cylinders, where $\varepsilon'=0.99\varepsilon$. An $\varepsilon'$-thick cylinder with waist $\gamma$ satisfies  $2\varepsilon'\leq\ell(\gamma)\leq 2\varepsilon$, since $\gamma$ is contained in the $\varepsilon$-thin part. An $\varepsilon'$-thin cylinder satisfies $\ell(\gamma)<2\varepsilon'$.

Lemma~\ref{lem:cylindernottooshort} in Section~\ref{subsec:thickhyperbolicsurfaces} states that the triangulation depicted in Figure~\ref{fig:standardtriangulation} is a Delaunay triangulation for $\varepsilon'$-thin cylinders. We call this triangulation a \emph{standard triangulation} and describe it in more detail in the following definition. For $\varepsilon'$-thick cylinders we use a different construction defined in Definition~\ref{def:standardcycle}.

\begin{definition}\label{def:triangulatingcylinder}
	Let $\mM$ be a closed hyperbolic surface. Let $C$ be an $\varepsilon'$-thin hyperbolic cylinder in $\mMthin$ with waist $\gamma$ and boundary curves $\gamma^+,\gamma^-$. Place three equally-spaced points $x_i,i=1,2,3$ on $\gamma$ (see Figure~\ref{fig:standardtriangulation}). Then, place three points $x_i^+,i=1,2,3$ on $\gamma^+$ and three points $x_i^-,i=1,2,3$ on $\gamma^-$ such that the projection of $x_i^\pm$ on $\gamma$ is equal to $x_i$ for $i=1,2,3$. Let $V$ be the set consisting of $x_i, x_i^-$ and $x_i^+$ for $i=1,2,3$. Let $E$ be the set of edges of one of the forms $$(x_i^-,x_{i+1}^-),(x_i^-,x_i),(x_i^-,x_{i+1}),(x_i,x_{i+1}),(x_i,x_i^+),(x_i,x_{i+1}^+),(x_i^+,x_{i+1}^+)$$
	for $i=1,2,3$ (counting modulo 3), where the embedding of an edge in $C$ is as shown in Figure~\ref{fig:standardtriangulation}. We call $(V,E)$ a \emph{standard triangulation} of $C$. 
\end{definition}  

\begin{figure}[h]
	\centering
	\includegraphics[width=0.35\textwidth]{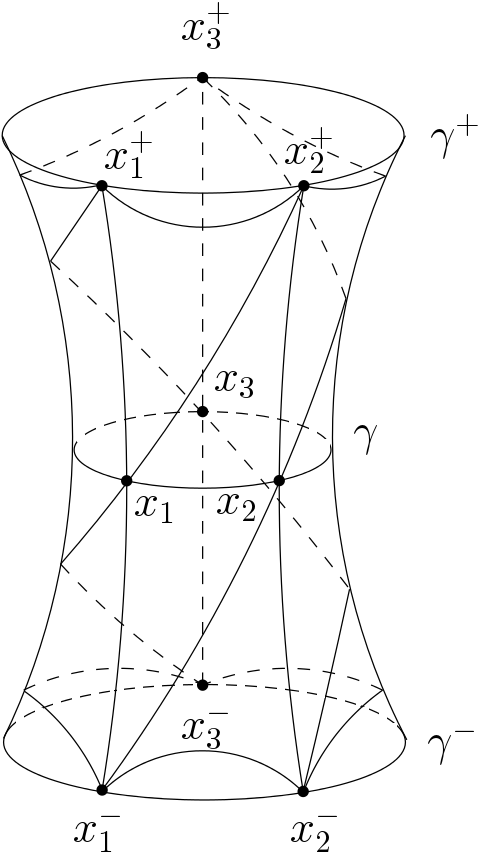}
	\caption{Standard triangulation of an $\varepsilon'$-thin cylinder.}
	\label{fig:standardtriangulation}
\end{figure}

We not only have to prove that a standard triangulation of an $\varepsilon'$-thin cylinder is a Delaunay triangulation, we also have to show that its edges are distance paths. Corollary~\ref{cor:cylinderdistancepaths} states that all edges in a standard triangulation are distance paths if $\varepsilon\leq 0.72$. Before we can prove Corollary~\ref{cor:cylinderdistancepaths}, we first need the following lemma.

\begin{lemma}\label{lem:bandaroundcollar}
	Let $\mM$ be a closed hyperbolic surface and let $\varepsilon\leq 0.72$. For each pair of distinct closed geodesics $\gamma_1$ and $\gamma_2$ in $\mMthin$ the collars $C_{\gamma_1}(K_{C_1}+\tfrac{1}{3}\varepsilon)$ and $C_{\gamma_2}(K_{C_2}+\tfrac{1}{3}\varepsilon)$ are embedded and disjoint.
\end{lemma}

\begin{remark}
	The value $0.72$ was found experimentally and is optimal up to two decimal digits, i.e., the statement is not true for $\varepsilon=0.73$. More specifically, if $\varepsilon\geq 0.73$ then there exists a closed hyperbolic surface $\mM$ with disjoint closed geodesics $\gamma_1$ and $\gamma_2$ in $\mMthin$ such that $C_{\gamma_1}(K_{C_1}+\tfrac{1}{3}\varepsilon)$ and $C_{\gamma_2}(K_{C_2}+\tfrac{1}{3}\varepsilon)$ are not disjoint.
\end{remark}

\begin{proof}
	See Figure~\ref{fig:cylindernesting}. We will show that $w(\gamma_i)-K_{C_i}\geq\tfrac{1}{3}\varepsilon$ for $i=1,2$. Namely, this implies that $C_{\gamma_i}(K_{C_i}+\tfrac{1}{3}\varepsilon)\subseteq C_{\gamma_i}(w(\gamma_i))$. Because $C_{\gamma_1}(w(\gamma_1))$ and $C_{\gamma_2}(w(\gamma_2))$ are embedded and disjoint by the Collar Lemma, it follows that $C_{\gamma_1}(K_{C_1}+\tfrac{1}{3}\varepsilon)$ and $ C_{\gamma_2}(K_{C_2}+\tfrac{1}{3}\varepsilon)$ are embedded and disjoint as well. 
	
	\begin{figure}[h]
		\centering
		\includegraphics[width=0.35\textwidth]{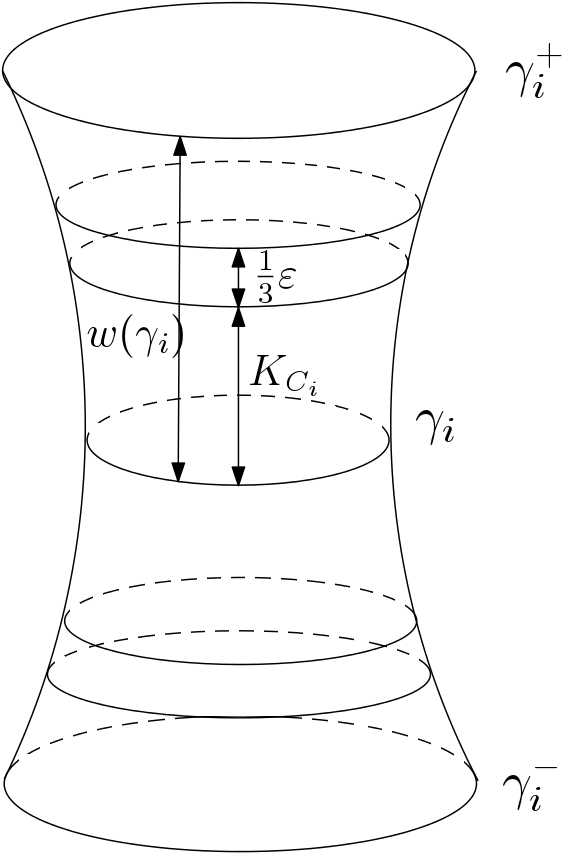}
		\caption{Illustration of the collars $C_{\gamma_i}(K_{C_i})\subset C_{\gamma_i}(K_{C_i}+\tfrac{1}{4}\varepsilon)\subseteq C_{\gamma_i}(w(\gamma_i))$.}
		\label{fig:cylindernesting}
	\end{figure}
	
	Comparing expression \eqref{eq:expressionKC} for $K_{C_i}$ and expression \eqref{eq:collarwidth} for $w(\gamma_i)$, we see that $w(\gamma_i)-K_{C_i}$ is a positive number, with infimum when $\ell(\gamma_i)\rightarrow 0$ \cite{parlier2016}. A straightforward computation shows that for $\varepsilon=0.72$ this infimum is equal to $0.24\ldots>\tfrac{1}{3}\varepsilon$. Since $w(\gamma_i)-K_{C_i}$ is decreasing as a function of $\varepsilon$, it follows that $w(\gamma_i)-K_{C_i}\geq \tfrac{1}{2}\varepsilon$ for all $\varepsilon\leq 0.72$. 
\end{proof}

\begin{corollary}\label{cor:cylinderdistancepaths}
	Let $\mM$ be a closed hyperbolic surface and let $\varepsilon\leq 0.72$. All edges in a standard triangulation of an $\varepsilon'$-thin cylinder in $\mMthin$ are distance paths. 
\end{corollary}

\begin{proof}
	It is clear that edges of the form $(x_i^-,x_{i+1}^-),(x_i,x_{i+1}),(x_i^+,x_{i+1}^+)$ for $i=1,2,3$ are distance paths. Now, consider the edge of length $K_C$ between $x_i$ and $x_i^+$. Because we know the metric of the cylinder, it can be shown explicitly that there are no shorter paths completely contained in the cylinder. Furthermore, because the collar $C_\gamma(K_C+\tfrac{1}{3}\varepsilon)$ is embedded by Lemma~\ref{lem:bandaroundcollar}, any path that leaves the top half of the cylinder and returns through the bottom half has length at least $K_C+\tfrac{2}{3}\varepsilon$. It follows that the edges of the form $(x_i,x_i^+)$ are distance paths. By symmetry, the edges of the form $(x_i^-,x_i)$ are distance paths as well.  
	
	Finally, consider the edge between $x_i$ and $x_{i+1}^+$. Because $d(x_i,x_{i+1})=\tfrac{1}{3}\ell(\gamma)<\tfrac{2}{3}\varepsilon$ and $d(x_{i+1},x_{i+1}^+)=K_C$, we see from the triangle inequality that $d(x_i,x_{i+1}^+)<K_C+\tfrac{2}{3}\varepsilon$. Because any path that leaves the top half of the cylinder and returns through the bottom part of the cylinder has length at least $K_C+\tfrac{2}{3}\varepsilon$ by the same reasoning as above, it follows that edges of the form $(x_i,x_{i+1}^+)$ are distance paths. By symmetry, edges of the form $(x_i^-,x_{i+1})$ are distance paths as well. 
\end{proof}

For $\varepsilon'$-thick cylinders, we see from Equation \eqref{eq:expressionKC} for $K_C$ that the width $K_C$ is close to zero. It turns out that we do not need to place three points on its waist and on each of its two boundary curves. Instead, three vertices on its waist suffice. 

\begin{definition}\label{def:standardcycle}
	Let $\mM$ be a closed hyperbolic surface. Let $C$ be a $\varepsilon'$-thick hyperbolic cylinder in $\mMthin$ with waist $\gamma$. Place three equally-spaced points $x_i,i=1,2,3$ on $\gamma$. Let $V=\{x_i\;|\;i=1,2,3 \}$ and $E=\{(x_1,x_2),(x_2,x_3),(x_3,x_1)\}$. We call $(V,E)$ a \emph{standard cycle} of $C$.
\end{definition}

\subsection{Constructing a distance Delaunay triangulation of $\mM$ with few vertices}\label{subsec:thickhyperbolicsurfaces}

After constructing sets of vertices in the cylinders in the $\varepsilon$-thin part $\mMthin$, we construct a sufficiently dense and well-distributed set of vertices in the remainder of the surface. The following definition shows more precisely how we construct a set of vertices in $\mMthick$ and a corresponding Delaunay triangulation. 

\begin{definition}\label{def:standardtriangulation}
	Set $\varepsilon=0.72$ and $\varepsilon'=0.99\varepsilon$. Let $\mM$ be a closed hyperbolic surface. Let $\mcP_1$ be the set consisting of the vertices of a standard triangulation of every $\varepsilon'$-thin cylinder in $\mMthin$ together with the vertices of a standard cycle for every $\varepsilon'$-thick cylinder in $\mMthin$. Let $T_j$ be the union of triangles in a standard triangulation $(V_j,E_j)$ of an $\varepsilon'$-thin cylinder $C_j$. For every $\varepsilon'$-thick cylinder $C_j$, set $T_j=\emptyset$. Define $\mcP_2$ to be a maximal set in $\mM\setminus\cup_j T_j$ such that $d(p,q)\geq\tfrac{1}{2}\varepsilon$ for all distinct $p\in\mcP_1\cup\mcP_2,q\in\mcP_2$. Denote the union $\mcP_1\cup\mcP_2$ by $\mcP$ and let $\mcT$ be the Delaunay triangulation of $\mcP$ on $\mM$ obtained after projecting a Delaunay triangulation of $\pi^{-1}(\mcP)$ in $\mD$ to $\mM$. We call $\mcT$ a \emph{thick-thin Delaunay triangulation} of $\mM$. The vertices in $\mcP_1$ and $\mcP_2$ are called the \emph{cylinder vertices} and \emph{non-cylinder vertices} of $\mcT$, respectively. 
\end{definition}

\begin{remark}
	Because by Corollary~\ref{cor:cylinderdistancepaths} all edges in a standard triangulation of any $\varepsilon'$-thin cylinder are distance paths if we choose $\varepsilon\leq0.72$, we have chosen $\varepsilon=0.72$ in Definition~\ref{def:standardtriangulation}. Namely, we will see in the proof of Theorem~\ref{thm:upperboundDT} that the larger we choose $\varepsilon$, the smaller the constant (in our case 151) in the upper bound for the number of vertices. As in Section~\ref{subsec:thinhyperboliccylinders} we will fix $\varepsilon=0.72$ and $\varepsilon'=0.99\varepsilon$ throughout this subsection. 
\end{remark}

The edges between vertices on the same boundary curve of $C_j$ are not equal to the boundary curves of $C_j$ (because the latter are not geodesics), so $T_j$ is strictly contained in $C_j$. We define $\mcP_2$ as a point set in $\mM\setminus\cup_j T_j$ instead of in $\mM\setminus\cup_j C_j$ to simplify our proof of Lemma~\ref{lem:simplicialcomplex}, where we show that a thick-thin Delaunay triangulation of $\mM$ is a simplicial complex.   

The definition of $\mcP$ does not explicitly forbid placing vertices of $\mcP_2$ in $\varepsilon'$-thick cylinders. However, we will see in the next lemma that there are no vertices of $\mcP_2$ in $\varepsilon'$-thick cylinders, because then they would be too close to the vertices of a standard cycle. 

\begin{lemma}\label{lem:standardcycle}
	Let $\mM$ be a closed hyperbolic surface and let $\mcT$ be a thick-thin Delaunay triangulation of $\mM$. Every vertex of $\mcT$ contained in an $\varepsilon'$-thick cylinder in $\mMthin$ is a cylinder vertex.  
\end{lemma} 

\begin{proof}
	Let $\mcP_1$ be the set of cylinder vertices and $\mcP_2$ the set of non-cylinder vertices. Let $C$ be an arbitrary $\varepsilon'$-thick cylinder with waist $\gamma$ and standard cycle $(V,E)$. We will show that the union $U$ of the disks of radius $\tfrac{1}{2}\varepsilon$ centered at the vertices of $V$ covers $C$ completely. Namely, this implies that every point of $C$ has distance at most $\tfrac{1}{2}\varepsilon$ to a vertex of $V$. Because $d(p,q)\geq\tfrac{1}{2}\varepsilon$ for all $p\in\mcP_1$ and $q\in\mcP_2$, it follows that there are no vertices of $\mcP_2$ contained in $C$.
	
	To prove that $U$ covers $C$ completely, first observe that $d(x_i,x_{i+1})=\tfrac{1}{3}\ell(\gamma)<\tfrac{2}{3}\varepsilon$ for all $i=1,2,3$ (counting modulo 3). Therefore, the circles of radius $\tfrac{1}{2}\varepsilon$ centered at $x_i$ and $x_{i+1}$ intersect in two points, of which we call one $p$. Since the collar $C_\gamma(d(\gamma,p))$ is contained in $U$, it suffices to show that $K_C<d(\gamma,p)$, because then $C=C_\gamma(K_C)\subset C_\gamma(d(\gamma,p))\subset U$. From equation \eqref{eq:expressionKC} for $K_C$ we know that
	$$ \cosh(K_C)=\dfrac{\sinh(\varepsilon)}{\sinh(\tfrac{1}{2}\ell(\gamma))}\leq\dfrac{\sinh(\varepsilon)}{\sinh(\varepsilon')}\leq 1.02,$$
	where we substituted $\varepsilon'=0.99\varepsilon$ and $\varepsilon=0.72$ in the last step. On the other hand, the hyperbolic Pythagorean theorem yields
	$$ \cosh (d(\gamma,p))=\dfrac{\cosh(\tfrac{1}{2}\varepsilon)}{\cosh(\tfrac{1}{6}\ell(\gamma))}\geq\dfrac{\cosh(\tfrac{1}{2}\varepsilon)}{\cosh(\tfrac{1}{3}\varepsilon)}\geq 1.03,$$
	(see Figure~\ref{fig:standardcycletriangle}) where again we substituted $\varepsilon=0.72$ in the last step. We conclude that $K_C<d(\gamma,p)$, which finishes the proof.
	
	\begin{figure}[h]
		\centering
		\includegraphics[width=0.5\textwidth]{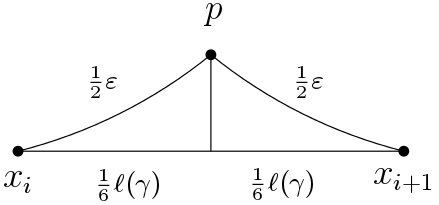}
		\caption{Computing $d(\gamma,p)$.}
		\label{fig:standardcycletriangle}
	\end{figure}
	
\end{proof}

Even though the set of vertices of a thick-thin Delaunay triangulation of $\mM$ contains the vertices of a standard triangulation $(V_j,E_j)$ for every $\varepsilon'$-thin cylinder $C_j$, a piori it is not clear that the edges in $E_j$ are edges in $\mcT$ as well. In the next lemma, we will show that for every $\varepsilon'$-thin cylinder the triangles in a standard triangulation are Delaunay triangles with respect to the set of vertices of any thick-thin Delaunay triangulation of $\mM$. Namely, if this holds, then there exists a Delaunay triangulation of $\mcP$ on $\mM$ containing a standard triangulation of every $\varepsilon'$-thin cylinder in $\mMthin$.

\begin{lemma}\label{lem:cylindernottooshort}
	Let $\mM$ be a closed hyperbolic surface. Let $\mcT$ be a thick-thin Delaunay triangulation of $\mM$ with vertex set $\mcP$ and let $C$ be an $\varepsilon'$-thin cylinder in $\mMthin$ with waist $\gamma$. Let $(V,E)$ be a standard triangulation of $C$ such that $V\subset\mcP$. Then all triangles of $(V,E)$ are Delaunay triangles with respect to the point set $\mcP$.
\end{lemma}

\begin{remark}\label{rem:distanceverticesstandardtriangulation}
	The proof of Lemma~\ref{lem:cylindernottooshort} is given in Appendix~\ref{sec:appendixupperbound}. Even though we do not give the proof here, we note that in the proof it is shown as an intermediate step that $d(x_i^\pm,x_{i+1}^\pm)<\varepsilon$ for all $i=1,2,3$. This inequality is used once more in the proof of Lemma~\ref{lem:graphwithshortedges}. 
\end{remark}
  
Henceforth, we will assume that for each $\varepsilon'$-thin cylinder the vertices and edges of a standard triangulation are contained in a thick-thin Delaunay triangulation of $\mM$. To show that $\mcT\in\mcD(\mM)$, we must show that $\mcT$ is a simplicial complex, i.e. it does not contain any 1- or 2-cycles, and that its edges are distance paths. 

In the next lemma, we show that any edge that intersects $\mMthick$ has length smaller than $\varepsilon$. Moreover, we show that it follows that all edges that intersect $\mMthick$ are distance paths and that there are no 1- and 2-cycles consisting of edges intersecting $\mMthick$. 

\begin{lemma}\label{lem:graphwithshortedges}
	Let $\mM$ be a closed hyperbolic surface and let $\mcT$ be a thick-thin Delaunay triangulation of $\mM$. Any edge of $\mcT$ that intersects $\mMthick$ has length smaller than $\varepsilon$ and is a distance path. Moreover, there are no 1- or 2-cycles that intersect $\mMthick$ and consist of edges of length smaller than $\varepsilon$.
\end{lemma}

\begin{proof}
	Let $(u,v)$ be an edge of $\mcT$ with non-empty intersection with $\mMthick$. Assume that $(u,v)$ is contained in a triangle $(u,v,w)$ in $\mcT$ with circumradius $r$ and circumcenter $c$. We will first show that $\ell(u,v)<\varepsilon$. We consider two cases, depending on which set $c$ is contained in. First, if $c\in T_j$ for some $\varepsilon'$-thin cylinder $C_j$, then at least one of $u$ and $v$ is contained in $\mcP_1$. If both are contained in $\mcP_1$, then $(u,v)$ is contained in $T_j$, because the edges of a standard triangulation of an $\varepsilon'$-thin cylinder are distance paths by Corollary~\ref{cor:cylinderdistancepaths}. This contradicts $(u,v)\cap\mMthick\neq\emptyset$, so we can assume that only one of $u$ and $v$, say $v$, is contained in $\mcP_1$. Then without loss of generality the situation is as depicted in Figure~\ref{fig:cylindercircumdisk}, where $\{v,w\}=\{x_i^+,x_{i+1}^+\}$.
	
	\begin{figure}[h]
		\centering
		\includegraphics[width=0.4\textwidth]{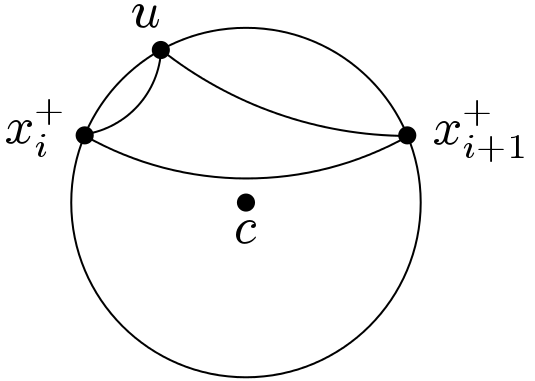}
		\caption{Circumscribed disk of a triangle $(u,x_i^+,x_{i+1}^+)$.}
		\label{fig:cylindercircumdisk}
	\end{figure}
	
	It follows that $\ell(u,v)<d(x_i^+,x_{i+1}^+)$. Because $d(x_i^+,x_{i+1}^+)<\varepsilon$ by Remark~\ref{rem:distanceverticesstandardtriangulation}, it follows that $\ell(u,v)<\varepsilon$.
	Second, if $c\in\mM\setminus\cup_{j\in I}T_j$, then we can deduce that $r<\tfrac{1}{2}\varepsilon$. Namely, if we suppose for a contradiction that $r\geq\tfrac{1}{2}\varepsilon$, then $d(c,p)\geq\tfrac{1}{2}\varepsilon$ for all $p\in\mcP$, because the circumcircle of $(u,v,w)$ is empty. Then we could add $c$ to $\mcP_2$, which contradicts its maximality. We conclude that $r<\tfrac{1}{2}\varepsilon$. Because $(u,v)$ is contained in a circle of radius $r<\tfrac{1}{2}\varepsilon$, it follows that $\ell(u,v)<\varepsilon$. Because $\ell(u,v)<\varepsilon$ in both cases, the first claim of the lemma follows.	
	
	To show that $(u,v)$ is a shortest distant path between its endpoints, suppose for a contradiction that it is not. Then there exists a geodesic $\gamma$ from $u$ to $v$, such that $\ell(\gamma)<\ell(u,v)$. This means that $(u,v)\cup \gamma$ is a homotopically non-trivial closed curve of length smaller than $2\ell(e)<2\varepsilon$. However, because $\injrad(x)>\varepsilon$ for all $x\in\mMthick$, every homotopically non-trivial closed curve $\gamma$ intersecting $\mMthick$ has length at least $2\varepsilon$, which contradicts $\ell((u,v)\cup\gamma)<2\varepsilon$. We conclude that $(u,v)$ is a distance path between its endpoints.
	
	A 1- or 2-cycle in $G$ corresponds to a homotopically non-trivial closed curve on $\mM$ \cite{bogdanov2016}. By the same argument as before, the length of a 1- or 2-cycle $\sigma$ intersecting $\mMthick$ is at least $2\varepsilon$. Therefore, there are no 1- or 2-cycles that intersect $\mMthick$ and consist of edges of length smaller than $\varepsilon$. 
\end{proof}

Using the previous lemma, we show that a thick-thin Delaunay triangulation of $\mM$ is a distance Delaunay triangulation. 

\begin{lemma}\label{lem:simplicialcomplex}
	Every thick-thin Delaunay triangulation of a closed hyperbolic surface is a distance Delaunay triangulation. 
\end{lemma}

\begin{proof}
	Let $\mM$ be a closed hyperbolic surface and let $\mcT$ be a thick-thin Delaunay triangulation of $\mM$. By definition, $\mcT$ is a Delaunay triangulation. We will show that $\mcT$ does not contain any 1- or 2-cycles to prove that it is a simplicial complex. We know from Lemma~\ref{lem:graphwithshortedges} that any edge $(u,v)$ such that $(u,v)\cap\mMthick\neq\emptyset$ is not a 1-cycle. Because by construction there are no 1-cycles in a standard triangulation or standard cycle in $\mMthin$ as well, we conclude that $\mcT$ contains no 1-cycles.
	
	To prove that $\mcT$ does not contain any 2-cycles, consider two distinct edges $(u,v)$ and $(v,w)$ of $\mcT$ with at least one shared endpoint. There are three cases, depending on whether two, one or zero of the edges $(u,v)$ and $(v,w)$ intersect $\mMthick$.\\ First, if $(u,v)$ and $(v,w)$ both intersect $\mMthick$, then they do not form a 2-cycle by Lemma~\ref{lem:graphwithshortedges}.\\
	Second, if precisely one of $(u,v)$ and $(v,w)$, say $(u,v)$, intersects $\mMthick$, then $\ell(u,v)<\varepsilon$ and $(v,w)$ is an edge contained in a hyperbolic cylinder of $\mMthin$. If $(v,w)$ is an edge contained in an $\varepsilon'$-thick cylinder $C$ with waist $\gamma$, then $(v,w)$ is one of the edges of the standard cycle of $C$, because there are no other vertices in $C$ by Lemma~\ref{lem:standardcycle}. Then $\ell(v,w)=\tfrac{1}{3}\ell(\gamma)<\tfrac{2}{3}\varepsilon$, so $(u,v)$ and $(v,w)$ do not form a 2-cycle by Lemma~\ref{lem:graphwithshortedges}. Next, assume that $(v,w)$ is an edge in an $\varepsilon'$-thin cylinder with waist $\gamma$. Then either $w$ lies on $\gamma$ and $v$ lies on one of the boundary curves of $C$ or $v$ and $w$ both lie on the same boundary curve of $C$. If $w$ lies on $\gamma$ and $v$ on a boundary curve of $C$, then $(u,v)$ and $(v,w)$ do not form a 2-cycle, because $u$ does not lie on $\gamma$. If $v$ and $w$ both lie on the same boundary geodesic, then $\ell(v,w)<\varepsilon$ by Remark~\ref{rem:distanceverticesstandardtriangulation}, so $(u,v)$ and $(v,w)$ do not form a 2-cycle by Lemma~\ref{lem:graphwithshortedges}.\\
	Third, if neither $(u,v)$ nor $(v,w)$ intersects $\mMthick$, then $(u,v)$ and $(v,w)$ are both contained in a hyperbolic cylinder in $\mMthin$. They are contained in the same cylinder, because different cylinders are separated by $\mMthick$. Because by construction standard triangulations and standard cycles do not contain any 2-cycle, $(u,v)$ and $(v,w)$ do not form a 2-cycle.\\ 
	This finishes the case analysis and we conclude that $\mcT$ is a simplicial complex.
	
	To prove that all edges of $\mcT$ are distance paths, we know from Lemma~\ref{lem:graphwithshortedges} that any edge that intersects $\mMthick$ is a distance path. Because all edges in a standard triangulation are distance paths by Corollary~\ref{cor:cylinderdistancepaths} and because all edges in a standard cycle are distance paths by construction, we conclude that all edges in $\mcT$ are distance paths.
\end{proof}

\subsection{Proof of Theorem~\ref{thm:upperboundDT}}\label{subsec:prooftheorem1}

\begin{proof}
	\textbf{(Theorem~\ref{thm:upperboundDT})}\\
	Let $\mM$ be an arbitrary hyperbolic surface of genus $g$ and let $\mcT$ be a thick-thin Delaunay triangulation of $\mM$. By definition, $\mcT$ is a Delaunay triangulation. By Lemma~\ref{lem:simplicialcomplex}, $\mcT$ is a simplicial complex and all edges of $\mcT$ are distance paths. Hence, $\mcT\in\mcD(\mM)$.
	
	We will show here that the number of vertices of $\mcT$ is smaller than $151g$. By Lemma~\ref{lem:collarlemma2}, $\mMthin$ consists of at most $3g-3$ cylinders and each of these cylinders contains either 9 vertices (if it is $\varepsilon'$-thin) or 3 vertices (if it is $\varepsilon'$-thick). Therefore, $|\mcP_1|\leq 27g-27$. 
	
	To find an upper bound for the cardinality of $\mcP_2$, observe that for distinct $p,q\in\mcP_2$ the disks $B_p(\tfrac{1}{4}\varepsilon)$ and $B_q(\tfrac{1}{4}\varepsilon)$ of radius $\tfrac{1}{4}\varepsilon$ centered at $p$ and $q$, respectively, are embedded and disjoint. Therefore, the cardinality of $\mcP_2$ is bounded above by the number of disjoint, embedded disks of radius $\tfrac{1}{4}\varepsilon$ that we can fit in $\mM$. Because the area of a hyperbolic disk of radius $\tfrac{1}{4}\varepsilon$ is $2\pi (\cosh(\tfrac{1}{4}\varepsilon) -1)$ \cite{beardon2012} and because the area of $\mM$ is $4\pi(g-1)$ \cite{stillwell2012}, we obtain 	
	$$ |\mcP_2|\leq \dfrac{2(g-1)}{\cosh(\tfrac{1}{4}\varepsilon)-1}.$$
	Therefore,   
	$$ |\mcP|\leq 27g-27+\dfrac{2(g-1)}{\cosh(\tfrac{1}{4}\varepsilon)-1}\leq 151g.$$
	This finishes the proof. 
\end{proof}

\begin{remark}
The constant 151 is not optimal. We can obtain the stronger upper bound $|\mcP|\leq 124 g$ by looking more precisely at the upper bounds of $|\mcP_1|$ and $|\mcP_2|$ but because we are mainly interested in the the order of growth, we will not provide any details. 
\end{remark}

\section{Classes of hyperbolic surfaces attaining the order of the upper bound}\label{sec:proofmaintheorem}\label{sec:upperboundtight}

As our second result, we show that there exists a class of hyperbolic surfaces which attains the order of the upper bound presented in Theorem~\ref{thm:upperboundDT}. We will first introduce this class of hyperbolic surfaces and then state the precise result in Theorem~\ref{thm:lowerbound}.

Recall from the Preliminaries that cutting a hyperbolic surface along $3g-3$ disjoint simple closed geodesics decomposes the surface into $2g-2$ pairs of pants and that each pair of pants decomposition has an associated 3-regular graph. Conversely, define $L_g$ as the trivalent graph depicted in Figure~\ref{fig:Ggandpopdecomp}. Here, every vertex $v_i$ corresponds to a pair of pants $Y_i$. There is one edge from $v_1$ to itself and similarly from $v_{2g-2}$ to itself. Moreover, for $1\leq i\leq 2g-3$ there is one edge between $v_i$ and $v_{i+1}$ if $i$ is odd and there are two edges if $i$ is even. 

\begin{figure}[h]
	\begin{subfigure}{\textwidth}
		\centering
		\includegraphics[width=0.9\textwidth]{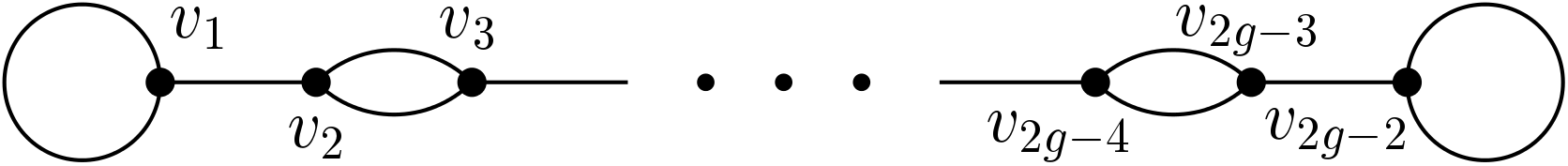}
		\caption{Trivalent graph $L_g$.}
		\label{fig:defGg}
	\end{subfigure}
	\begin{subfigure}{\textwidth}
		\centering
		\includegraphics[angle=-90, width=\textwidth]{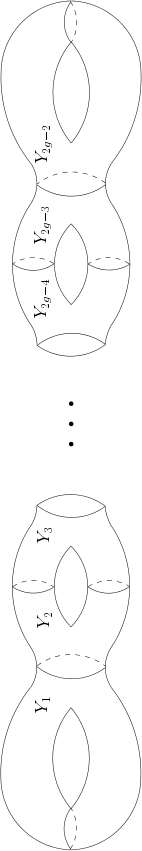}
		\caption{Pair of pants decomposition corresponding to $L_g$}
		\label{fig:popdecompGg}
	\end{subfigure}
	\caption{Trivalent graph $L_g$ with corresponding pair of pants decomposition.}
	\label{fig:Ggandpopdecomp}
\end{figure}

Now, fix some interval $[a,b]\subset\mR$ with $0<a<b$. Let $S_g(a,b)$ be the subset of $\teich_g$ with underlying graph $L_g$ such that all length parameters are contained in $[a,b]$. In particular, $S_g(a,b)$ contains an open subset of $\teich_g$, showing that having a linear number of vertices in terms of genus is relatively stable in this part of Teichm\"uller space. We will now state the result of this section.

\begin{theorem}\label{thm:lowerbound}
There exists a constant $B>0$ depending only on $a,b$ such that a minimal distance Delaunay triangulation of any hyperbolic surface in $S_g(a,b)$ has at least $Bg$ vertices.
\end{theorem}

The idea of the proof is the following. Let a hyperbolic surface $\mM\in S_g(a,b)$ and $\mcT\in\mcD(\mM)$ be given. Euler's formula implies $v-\tfrac{1}{3}e=2-2g$ for triangulations of a surface of genus $g$, where $v$ and $e$ are the number of vertices and edges of the triangulation. We prove that $e\leq B'v$ for some constant $B'>3$ only depending on $a,b$, which implies that
$$ v\geq \dfrac{6g-6}{B'-3}.$$
This implies the result of Theorem~\ref{thm:lowerbound}. Hence, the argument consists mostly in finding an upper bound for the number of edges in terms of the number of vertices. 

Before we continue with the proof of Theorem~\ref{thm:lowerbound}, we will look at our construction of $S_g(a,b)$ in more detail. By definition, every boundary geodesic of a pair of pants in the pair of pants decomposition of $\mM\in S_g(a,b)$ with respect to $L_g$ has length in $[a,b]$. As explained in Section~\ref{sec:preliminaries}, the geometry of a pair of pants depends continuously on the lengths of its three boundary geodesics. In particular, the diameter $\diam(Y)$ of a pair of pants $Y$ as well as the minimal distance $\mindist(Y)$ between any two of its boundary geodesics depend continuously on the lengths of its boundary geodesics. Because $[a,b]$ is a compact set, we obtain as an immediate consequence the following lemma.
\begin{lemma}
There exist positive numbers $m(a,b)$ and $M(a,b)$ depending on $a$ and $b$ such that $m(a,b)\leq\mindist(Y)<\diam(Y)\leq M(a,b)$ for every pair of pants $Y$ whose boundary geodesics have length in $[a,b]$. 
\end{lemma}

\begin{remark}
It is not too difficult to compute bounds for $\mindist(Y)$ and $\diam(Y)$ in terms of the lengths of the boundary geodesics of $Y$. This would give explicit expressions for $m(a,b)$ and $M(a,b)$ in terms of $a$ and $b$. As we are only interested in the order of growth, to avoid further technical details, we do not provide details.
\end{remark}

From now on, the numbers $m=m(a,b)$ and $M=M(a,b)$ will be fixed. Furthermore, the way in which we constructed the graph $L_g$ allows us to speak of \emph{consecutive} pairs of pants. A \emph{cluster} in a hyperbolic surface $\mM$ is a subset of $\mM$ consisting of a number of consecutive pairs of pants. Consider $\mcT\in\mcD(\mM)$. An \emph{$m$-gap} is a cluster consisting of $m$ consecutive empty pairs of pants, where empty means that the pairs of pants do not contain any vertices of $\mcT$. We say that an edge of $\mcT$ \emph{crosses} a cluster if the pairs of pants containing its endpoints are separated by the cluster. Note that the cluster need not contain all pairs of pants which separate the two endpoints.

The next lemma states that if an edge of a distance Delaunay triangulation crosses many pairs of pants, then ``many'' of these pairs of pants are empty. 

\begin{lemma}\label{lem:N-empty}
	Let $\mM\in S_g(a,b)$ and define $N=N(a,b)$ as $$N(a,b):=\bigg\lceil\dfrac{M(a,b)}{m(a,b)}\bigg\rceil+1.$$
	Then, for every $\mcT\in\mcD(\mM)$:
	\begin{enumerate}
		\item
		If an edge of $\mcT$ crosses a cluster consisting of at least $3N$ pairs of pants, this
		cluster contains an $N$-gap.
		\item
		If an edge of $\mcT$ crosses a cluster in which the first $N$ and the last $N$ pairs of pants are empty, then all pairs of pants in the cluster are empty.
	\end{enumerate}
\end{lemma}	

\begin{proof}
	Let $\mcT\in\mcD(\mM)$ and let $(u,v)$ be an edge of $\mcT$ with $u\in Y_i$ and $v\in Y_j$. We will show that the cluster consisting of the union of all $Y_k$ with $i+N+1\leq k\leq j-N-1$ is empty. In other words, only the first $N$ and last $N$ pairs of pants are possibly non-empty. In particular, this implies that the two properties of the lemma hold. 
	
	Now, because $(u,v)$ is a Delaunay edge, it is contained in some empty disk $D$ passing through $u$ and $v$. Consider $Y_k$ with $i+N+1\leq k\leq j-N-1$ and take $p\in Y_k$ arbitrarily (see Figure~\ref{fig:pairofpants}). We will show that the distance between $p$ and the center $c$ of $D$ satisfies $\dist(p,c)<\dist(u,c)$. This implies that $p$ is contained in the interior of $D$, so it cannot be a vertex of $\mcT$. Therefore, $Y_k$ is empty.
		
		\begin{figure}[h]
			\centering
			\includegraphics[width=\textwidth]{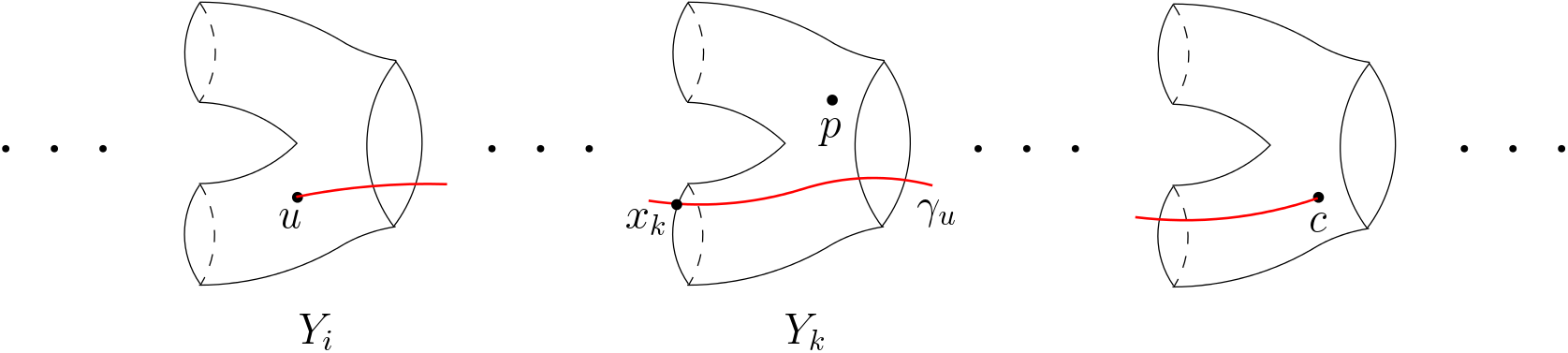}
			\caption{Depiction of the construction in the proof of Lemma~\ref{lem:N-empty}.}
			\label{fig:pairofpants}
		\end{figure}
		
	Let $\gamma_u$ be a distance path from $c$ to $u$. First assume that $\gamma_u\cap Y_k\neq\emptyset$. Let $x_k$ be the intersection of $\gamma_u$ with one of the boundary geodesics of $Y_k$. By the triangle inequality and the definition of $M$, we know that
	\[ \dist(c,p)\leq \dist(c,x_k)+\dist(x_k,p)\leq \dist(c,x_k)+M.\]
	To give an upper bound for $\dist(c,x_k)$, observe that the part of $\gamma_u$ from $x_k$ to $u$ passes through $Y_{k-1},\ldots,Y_{i+1}$ before reaching $u\in Y_i$. By definition of $m$, the length of the part of $\gamma_u$ within each of these $k-1-i$ pair of pants is at least $m$, so $\dist(x_k,u)\geq (k-1-i)m$. This means that 
	\[\dist(c,x_k)=\dist(c,u)-\dist(x_k,u)\leq \dist(c,u)-(k-1-i)m.\]
	It follows that
	\[\dist(c,p)\leq \dist(c,u)+M-(k-1-i)m. \]
	Because $k-1-i\geq N$ and $Nm> M$, we see that $M-(k-1-i)m<0$, so
	\[\dist(c,p)< \dist(c,u),  \]
	so $Y_k\in D$, hence $Y_k$ does not contain any vertices of $\mcT$. Note that we assumed that $\gamma_u\cap Y_k\neq\emptyset$. If $\gamma_u\cap Y_k=\emptyset$, then we consider a distance path $\gamma_v$ from $c$ to $v$ instead of $\gamma_u$. The rest of the proof is analogous. We conclude that $Y_k$ is empty for $i+N+1\leq k\leq j-N-1$. This finishes the proof.
\end{proof}

The following lemma states that we can construct a set of clusters which has as one of its properties that every edge of the distance Delaunay triangulation has its endpoints in the same cluster, or in two consecutive clusters. 

\begin{lemma}
	\label{lem:decomposition}
	Let $\mM\in S_g(a,b)$ be a hyperbolic surface and let $N=N(a,b)$ be as defined in Lemma~\ref{lem:N-empty}. Let $\mcT\in\mcD(\mM)$. There are interior-disjoint clusters with the following properties:
	\begin{enumerate}
		\item
		\label{it:size}
		Each cluster consists of at most $6N$ consecutive pairs of pants;
		\item
		\label{it:V}
		Every cluster contains at least one vertex of $\mathcal{T}$,
		and every vertex of $\mathcal{T}$ belongs to at least one cluster;
		\item
		\label{it:E}
		Every edge of $\mathcal{T}$ has its endpoints in the same cluster, or in two consecutive
		clusters.
	\end{enumerate}
\end{lemma}

\begin{proof}
	A \textit{wide gap} in the sequence of $2g-2$ pairs of pants is a maximal subsequence consisting of at least $N$ empty pairs of pants, where maximality is defined with respect to inclusion. The complement of the collection of wide gaps of $\mM$ consists of a number of sequences of consecutive pairs of pants, that we call \textit{superclusters} (see Figure~\ref{fig:decomposition}). To obtain clusters that satisfy the properties of the lemma, each supercluster will be chopped up into one or more subsequences of length at most $6N$ in the following way:
	\begin{itemize}
		\item 
		Each supercluster consisting of at most $3N$ pairs of pants is a cluster. Such a cluster is said to be a \emph{short cluster}.
		\item
		Each supercluster consisting of more than $3N$ pairs of pants is chopped up into non-overlapping
		subsequences of length $3N$, followed by a subsequence of length between $3N$ and $6N-1$.
		
		More precisely, let $m$ be the length of such a supercluster. Since $m > 3N$
		there are integers $k$ and $r$, with $0 \leq r < 3N$ and $k \geq 1$, such that $m = 3kN + r$.
		Let $m_1 = \cdots = m_{k-1} = 3N$ and $m_k = 3N + r$, then
		$m_1 + \cdots + m_k = m$ and $3N \leq m_i < 6N$.
		Therefore, the supercluster is the concatenation of $k$ subsequences of length $m_1,\ldots,m_k$.
	\end{itemize}
	This construction enforces Property~\ref{it:size}.
	
	\begin{figure}[h]
		\centering
		\includegraphics[width=\textwidth]{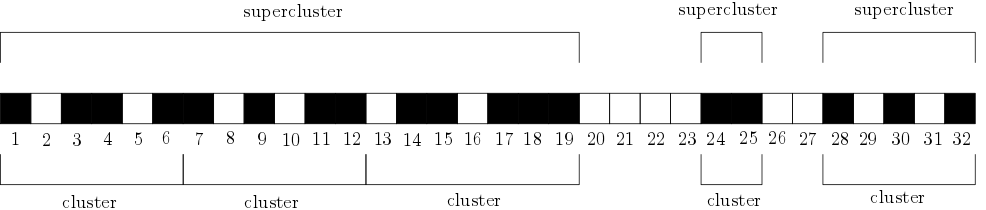}
		\caption{Illustration of the construction in the proof of Lemma~\ref{lem:decomposition}. Each square represents a pair of pants in the pair of pants decomposition of a hyperbolic surface of genus 17. A pair of pants corresponding to a black square contains vertices of $\mcT$, whereas a pair of pants corresponding to a white square does not contain vertices of $\mcT$. We assume that $N=2$. The superclusters and clusters that are defined during the construction are indicated.}
		\label{fig:decomposition}
	\end{figure}
		
	\medskip\noindent
	\textit{Proof of Property~\ref{it:V}.}
	Each supercluster contains at least one vertex.
	Therefore, short clusters contain at least one vertex.
	The clusters obtained by chopping up a supercluster have length at least $3N$. Since a
	supercluster contains no $N$-gap, each of these clusters contains at least one vertex.
	
	Since all vertices belong to some supercluster, they belong to at least one cluster.
	This completes the proof of Property~\ref{it:V}.
	
	\medskip\noindent
	\textit{Proof of Property~\ref{it:E}.}
	Suppose the property does not hold. Then there is an edge of $\mathcal{T}$ with vertices in
	non-adjacent clusters.
	In other words, this edge crosses an other cluster, say $C$.
	The construction of clusters implies that $C$ does not contain an $N$-gap.
	By Part 1 of Corollary~\ref{lem:N-empty} the cluster $C$ consists of less than $3N$ pairs of
	pants. Therefore, $C$ is a short cluster. Since this short cluster is neither
	the first nor the last in the sequence of superclusters, it is preceded by a wide gap and succeeded by a wide	gap. Since a wide gap contains an $N$-gap, Part 2 of Corollary~\ref{lem:N-empty} implies that cluster $C$
	is empty. This contradicts Property~\ref{it:V}. Therefore, Property~\ref{it:E} holds.
\end{proof}

In the following corollary, we denote the number of vertices of $\mcT\in\mcD(\mM)$ contained in a subset $U$ of $\mM$ by $v(U)$. Likewise, let $e(U,W)$ be the number of edges of $\mcT$ with one endpoint in $U\subset\mM$ and one endpoint in $W\subset\mM$.

\begin{corollary}\label{cor:boundverticesedges}
	Let $\mM\in S_g(a,b)$ be a hyperbolic surface and let $\mcT\in\mcD(\mM)$. Let $\{\mS_i\;|\;i=1,\ldots,n \}$ be a collection of clusters satisfying the properties of Lemma~\ref{lem:decomposition} for some $n\in\mN$. If $v$ and $e$ are the number of vertices and edges of $\mcT$, respectively, then
	\begin{align*}
		n&\leq v,\\
		v&=\sum_{i=1}^{n}v(\mS_i),\\
		e&=\sum_{i=1}^{n}e(\mS_i,\mS_i)+\sum_{i=1}^{n-1}e(\mS_i,\mS_{i+1}).
	\end{align*}
\end{corollary}

\begin{proof}
	Because every cluster contains at least one vertex, the number of clusters is at most the number of vertices, which proves the first equation. The second equation follows from the property that every vertex is contained in a cluster. Because every edge has its endpoints in the same cluster, or in two consecutive clusters, the third equation holds.
\end{proof}

Recall that we want to find a linear upper bound for the number of edges of a distance Delaunay triangulation in terms of the number of vertices. By Corollary~\ref{cor:boundverticesedges}, it suffices to find upper bounds for $e(\mS_i,\mS_i)$ and $e(\mS_i,\mS_{i+1})$ for clusters $\mS_i$ satisfying the properties of Lemma~\ref{lem:decomposition}. We will do this in the next lemma.

\begin{lemma}\label{lem:edgebound}
	Let the notation be as in Corollary~\ref{cor:boundverticesedges}. Then, the following upper bounds hold:
	\begin{enumerate}
		\item $e(\mS_i,\mS_i)\leq 3v(\mS_i)+18N(N+1)$ for all $i=1,\ldots,n$,
		\item $e(\mS_i,\mS_{i+1})\leq 18v(\mS_i\cup\mS_{i+1})+216N(N+1)$ for all $i=1,\ldots,i-1$.
	\end{enumerate}
\end{lemma}

The proof is given in Appendix~\ref{sec:appendixupperboundtight}. We can now commence with the proof of Theorem~\ref{thm:lowerbound}.

\begin{proof}
	\textbf{(Theorem~\ref{thm:lowerbound})}\\
	Take $\mM\in S_g(a,b)$ arbitrary and let $\mcT\in\mcD(\mM)$ be arbitrary. Let $\{\mS_i\;|\;i=1,\ldots,n\}$ be a collection of clusters satisfying the properties of Lemma~\ref{lem:decomposition}. By Corollary~\ref{cor:boundverticesedges}, 
	$$ e = \sum_{i=1}^{n}e(\mS_i,\mS_i)+\sum_{i=1}^{n-1}e(\mS_i,\mS_{i+1}).$$
	Substituting the upper bounds for $e(\mS_i,\mS_i)$ and $e(\mS_i,\mS_{i+1})$ from Lemma~\ref{lem:edgebound}, we obtain
	\begin{align*}
		e&\leq \sum_{i=1}^n \big(3v(\mS_i)+18N(N+1)\big)+\sum_{i=1}^{n-1}\big(18v(\mS_i\cup\mS_{i+1})+216N(N+1)\big),\\
		&\leq 39\sum_{i=1}^{n}\big(v(\mS_i)+6N(N+1)\big).
	\end{align*}
	From Corollary~\ref{cor:boundverticesedges}, we know that $\sum_{i=1}^{n}v(\mS_i)=v$ and $n\leq v$. Hence,
	$$ e\leq 39(1+6N(N+1))v.$$
	Euler's formula for triangulations $v-\tfrac{1}{3}e=2-2g$ implies that 
	\begin{align*}
		39(1+6N(N+1))v&\geq e=3v+6g-6,\\
		v&\geq \dfrac{g-1}{6+39N(N+1)},
	\end{align*}
	which finishes the proof.
\end{proof}

\section{Lower bound}\label{sec:lowerbound}

In this section, we will look at a general lower bound for the minimal number of vertices of a distance Delaunay triangulation of a hyperbolic surface of genus $g$. 

In the more general situation of a simplicial triangulation of a topological surface of genus $g$, one has an immediate lower bound on the minimal number of vertices. The fact that this lower bound is sharp is the following classical theorem of Jungerman and Ringel:
\begin{theorem}{\cite[Theorem 1.1]{jungerman1980}}\label{thm:minverticestopsurface}
The minimal number of vertices of a simplicial triangulation of a topological surface of genus $g$ is
	$$ \bigg\lceil\dfrac{7+\sqrt{1+48g}}{2}\bigg\rceil.$$
\end{theorem}

We show that  the same result holds for the minimal number of vertices of a distance Delaunay triangulation of a hyperbolic surface of genus $g$ for infinitely many values of $g$. 

\begin{theorem}\label{thm:generallowerbound}
For any $g\geq 2$ of the form
$$g=\dfrac{(n-3)(n-4)}{12}$$
for some $n\equiv 0\mod 12$, the minimal number of vertices of a distance Delaunay triangulation of a hyperbolic surface of genus $g$ is 
	$$ n=\dfrac{7+\sqrt{1+48g}}{2}.$$
\end{theorem}

\begin{proof}
	Because every distance Delaunay triangulation of a hyperbolic surface is a simplicial triangulation of the corresponding topological surface, it follows from Theorem~\ref{thm:minverticestopsurface} that the minimal number of vertices is at least  
	$$ \bigg\lceil\dfrac{7+\sqrt{1+48g}}{2}\bigg\rceil.$$
	In the remainder of the proof, we will construct for a given hyperbolic surface a distance Delaunay triangulation with the required number of vertices, inspired by a similar construction in the context of the chromatic number of hyperbolic surfaces \cite{parlier2016}. 
	
	Let $n\equiv 0\mod 12$ and assume that $n\neq 0$. The complete graph $K_n$ on $n$ vertices can be embedded in a topological surface $S_g$ of genus 
	$$ g=\dfrac{(n-3)(n-4)}{12},$$
	which is the smallest possible genus \cite{ringel1968}. Because we have assumed that $n\equiv 0\mod 12$, we know that the embedding of $K_n$ into $S_g$ is a triangulation $\mcT$ \cite{terry1967}. To turn $\mcT$ into a distance Delaunay triangulation, we will add a hyperbolic metric to the topological surface as follows. Every triangle in $\mcT$ is replaced by the unique equilateral hyperbolic triangle with all three angles equal to $\tfrac{2\pi}{n-1}$. In the complete graph $K_n$ every vertex has $n-1$ neighbouring vertices. This means that in every vertex $n-1$ equilateral triangles meet, so the total angle at each vertex is $2\pi$. Therefore, the result after replacing all triangles in $\mcT$ by hyperbolic triangles is a smooth hyperbolic surface $Z_g$. 
	
	It remains to be shown that $\mcT\in\mcD(Z)$. By construction, $\mcT$ is a simplicial complex. It has also been shown that all edges are distance paths \cite{parlier2016}. We will show here that $\mcT$ is a Delaunay triangulation of $Z_g$. Consider an arbitrary triangle $(u,v,w)$ in $\mcT$ with circumcenter $c$ and let $p\not\in\{u,v,w\}$ be an arbitrary vertex of $\mcT$ (Figure~\ref{fig:equilateraltriangles}). Consider a distance path $\gamma$ from $c$ to $p$. We can regard $\gamma$ as the concatenation of simple segments that each pass through an individual triangle. 
	
	\begin{figure}[h]
		\centering
		\includegraphics[width=.8\textwidth]{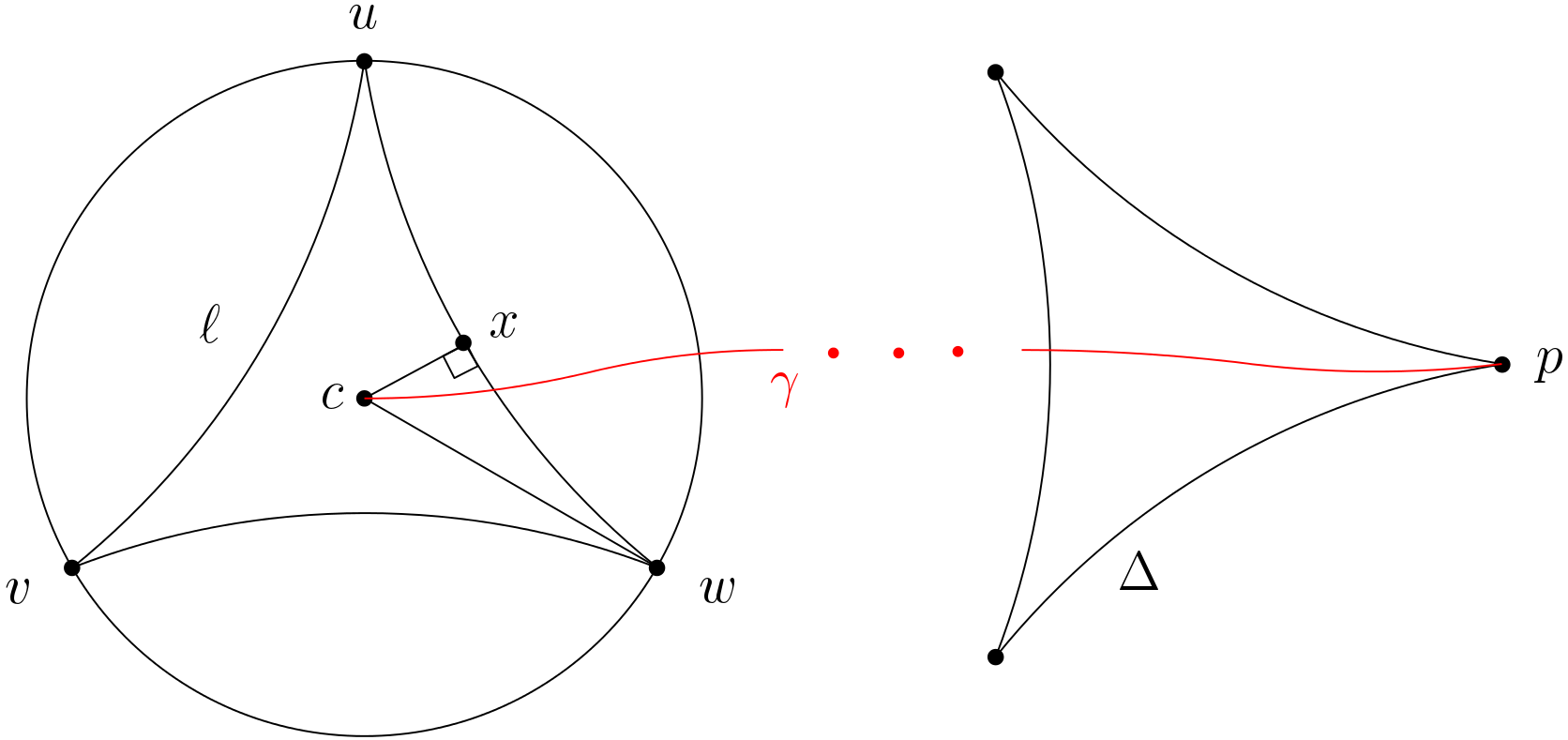}
		\caption{Schematic overview of the proof of $\mcT$ being a Delaunay triangulation.}
		\label{fig:equilateraltriangles}
	\end{figure}
	
	The first of these simple segments starts from $c$ and leaves the triangle $(u,v,w)$, so its length is at least the distance between $c$ and a side of $(u,v,w)$. Therefore, denoting by $x$ the projection of $c$ on one of the edges as shown in Figure~\ref{fig:equilateraltriangles}, the length of the first segment is at least $d(c,x)$. The last of the simple segments passes through a triangle, say $\Delta$, before arriving at $p$, so it has to pass through the side of $\Delta$ opposite to $p$. Therefore, its length is at least the distance between $p$ and the opposite side of $\Delta$. It is known that the distance between a vertex and the opposite side of an equilateral triangle is at least $\tfrac{1}{2}\ell$, where $\ell$ denotes the length of the sides of the equilateral triangle \cite{parlier2016}. Hence, $d(c,p)=\ell(\gamma)\geq d(c,x)+\tfrac{1}{2}\ell$. By the triangle inequality in triangle $(c,w,x)$ we see that $d(c,w)\leq d(c,x)+d(x,w)= d(c,x)+\tfrac{1}{2}\ell$, so we conclude that $d(c,p)\geq d(c,w)$. This means that $p$ is not contained in the interior of the circumcircle of $(u,v,w)$, which shows that $(u,v,w)$ is a Delaunay triangle. By symmetry, it follows that all triangles are Delaunay triangles, which finishes the proof.  
\end{proof}

\appendix

\section{Proof of Lemma~\ref{lem:cylindernottooshort}}\label{sec:appendixupperbound}

\begin{proof}
	\textbf{(Lemma~\ref{lem:cylindernottooshort})}\\
	To prove that the triangles of $(V,E)$ are Delaunay triangles, we will show that every circumscribed disk does not contain any point of $\mcP$ in its interior. By symmetry, it is sufficient to consider the top half of the cylinder. Let $i=1,2,3$ be arbitrary and denote the disk passing through $x_i^+, x_{i+1}^+, x_i, x_{i+1}$ by $D_i$. That $D_i$ does not contain any $p\in V$ in its interior is clear. The remainder of the proof consists of showing that $p$ is not contained in the interior of $D_i$ for all $p\in\mcP\setminus V$. Take $p\in\mcP\setminus V$ arbitrarily. Let $c_i$ be the center of $D_i$. If $d(c_i,p)>d(c_i,x_i)$, then $p$ is not contained in the interior of $D_i$.  
		
	Observe that $d(p,x_i^\pm)\geq\tfrac{1}{2}\varepsilon$ for $i=1,2,3$. Namely, if $p\in\mcP_2$, where $\mcP_2$ is the subset of $\mcP$ constructed in $\mMthick$, then by definition $d(p,x_i^\pm)\geq\tfrac{1}{2}\varepsilon$ for $i=1,2,3$. On the other hand, if $p\in\mcP_1$, then $p$ is a vertex in some hyperbolic cylinder $C'\neq C$ with waist $\gamma'$ in $\mMthin$. By Lemma~\ref{lem:bandaroundcollar}, the collars $C_\gamma(K_C+\tfrac{1}{3}\varepsilon)$ and $ C_{\gamma'}(K_{C'}+\tfrac{1}{2}\varepsilon)$ are disjoint, so the distance between $C$ and $C'$ is at least $\tfrac{2}{3}\varepsilon$. Hence, $d(p,x_i)\geq\tfrac{1}{2}\varepsilon$ for $i=1,2,3$. 
	
	\begin{figure}[h]
		\centering
		\includegraphics[width=0.85\textwidth]{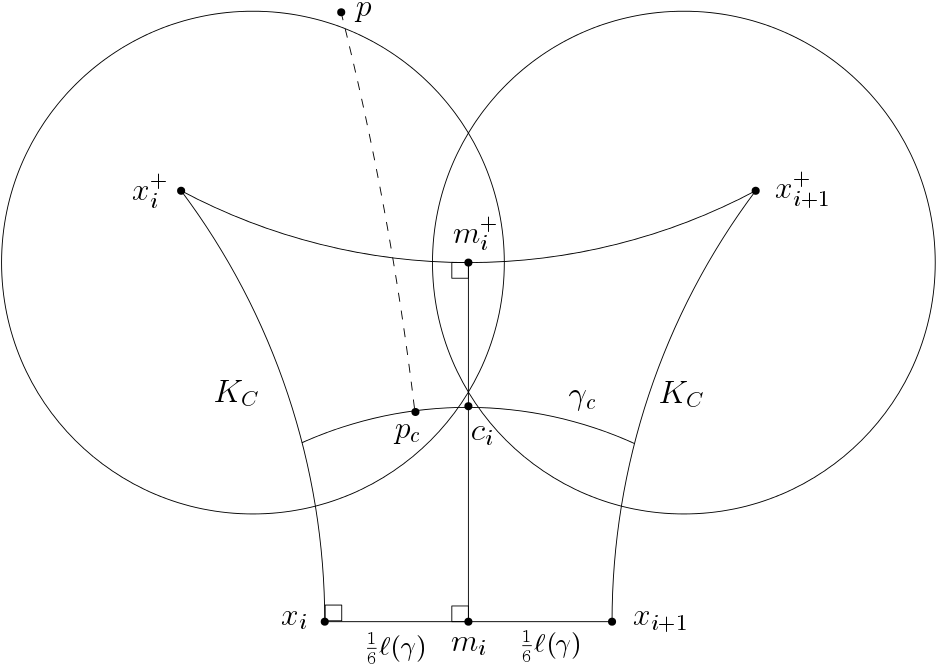}
		\caption{Construction to show that we can assume without loss of generality that $d(x_i^+,p)=d(x_{i+1}^+,p)=\tfrac{1}{2}\varepsilon$.
		\ref{lem:cylindernottooshort}.}
		\label{fig:cylinderdelaunaycircles}
	\end{figure}
	
	Now, we claim that we can assume without loss of generality that $d(x_i^+,p)=d(x_{i+1}^+,p)=\tfrac{1}{2}\varepsilon$. See Figure~\ref{fig:cylinderdelaunaycircles}. Consider the curve $\gamma_c$ consisting of points of distance $d(c_i,\gamma)$ from $\gamma$ and let $p_c$ be the point of $\gamma_c$ closest to $p$. Because $c_i\in\gamma_c$, we know that $d(p,c_i)\geq d(p,p_c)$. Now, note that a standard result from hyperbolic trigonometry in the quadrilateral $(x_i,x_i^+,m_i^+,m_i)$ with three right angles \cite[Formula Glossary 2.3.1(v)]{buser2010} states that
	$$ \sinh(\tfrac{1}{2}d(x_i^+,x_{i+1}^+))=\sinh(\tfrac{1}{6}\ell(\gamma))\cosh(K_C)=\dfrac{\sinh(\tfrac{1}{6}\ell(\gamma))\sinh(\varepsilon)}{\sinh(\tfrac{1}{2}\ell(\gamma))},$$
	where the last equality follows from expression \eqref{eq:expressionKC} for $K_C$. It can be deduced that $d(x_i^+,x_{i+1}^+)<\varepsilon$. Because $d(x,p)\geq \tfrac{1}{2}\varepsilon$ for all $x\in V$, the circles of radius $\tfrac{1}{2}\varepsilon$ centered at $x_i^+$ and $x_{i+1}^+$ intersect in two points. We see that $d(p,p_c)$ is minimized when $d(x_i^+,p)=d(x_{i+1}^+,p)=\tfrac{1}{2}\varepsilon$. Furthermore, if $d(x_i^+,p)=d(x_{i+1}^+,p)=\tfrac{1}{2}\varepsilon$, then $p$ lies on the geodesic passing through $m_i$ and $m_i^+$, so $d(p,p_c)=d(p,c_i)$, which means that $d(p,c)$ is minimized as well. We conclude that we can assume without loss of generality that $d(x_i^+,p)=d(x_{i+1}^+,p)=\tfrac{1}{2}\varepsilon$.
	
	\begin{figure}[h]
		\centering
		\includegraphics[width=0.7\textwidth]{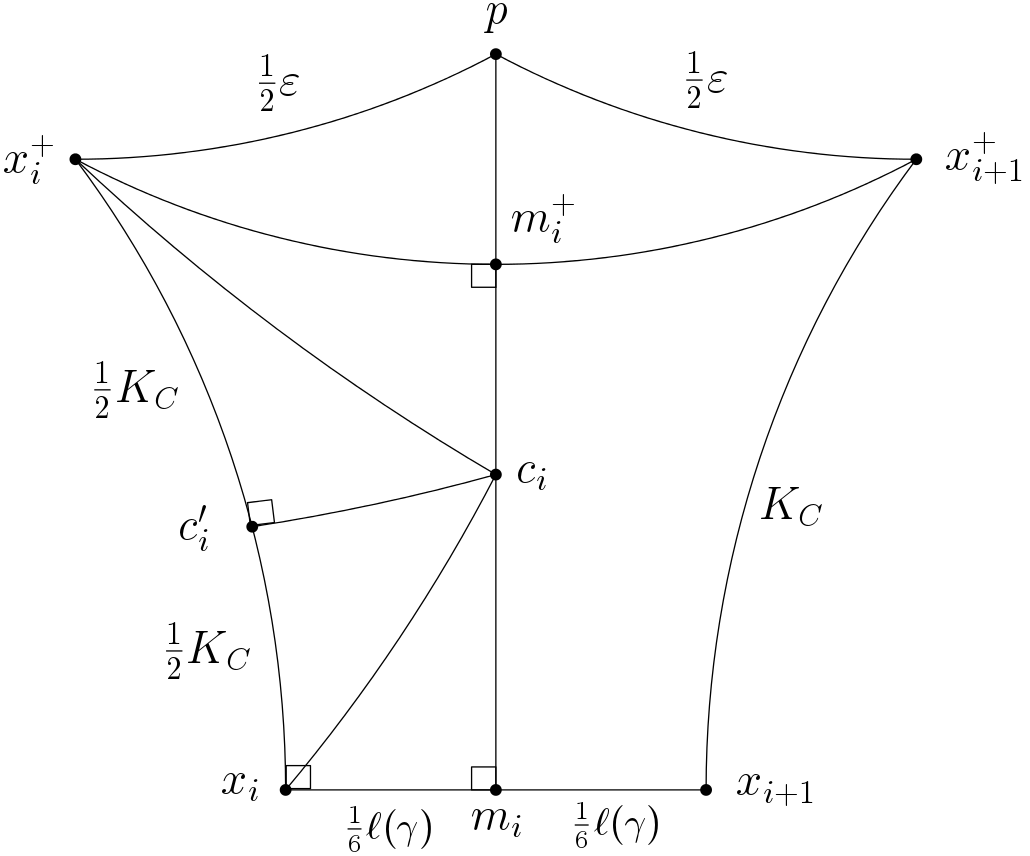}
		\caption{Schematic overview of the trigonometry in Lemma~\ref{lem:cylindernottooshort}.}
		\label{fig:cylinderdelaunay}
	\end{figure}
	
	See Figure~\ref{fig:cylinderdelaunay}, where $c_i'$ is the projection of $c_i$ on $(x_i,x_i^+)$. To prove that $$d(c_i,p)>d(c_i,x_i),$$
	observe that $d(c_i,p)=d(m_i,m_i^+)-d(m_i,c_i)+d(m_i^+,p)$, where $d(m_i,m_i^+),d(m_i,c_i)$ and $d(m_i^+,p)$ satisfy the equations
	\begin{align}
		\coth(d(m_i,m_i^+))&=\dfrac{\cosh(\tfrac{1}{6}\ell(\gamma))}{\tanh (K_C)},\label{eq:dmimiplus}\\
		\tanh(d(m_i,c_i))&=\dfrac{\cosh(\tfrac{1}{6}\ell(\gamma))}{\coth(\tfrac{1}{2}K_C)},\label{eq:dmici}\\
		\cosh(d(m_i^+,p))&=\dfrac{\cosh(\tfrac{1}{2}\varepsilon)}{\cosh(\tfrac{1}{2}d(x_i^+,x_{i+1}^+))}\label{eq:dmiplusp}.
	\end{align}
	Here, equation \eqref{eq:dmimiplus} follows from applying a standard formula in hyperbolic trigonometry \cite[Formula Glossary 2.3.1(iv)]{buser2010} in quadrilateral $(x_i,x_i^+,m_i^+,m_i)$. Equation \eqref{eq:dmici} follows from applying the same formula in quadrilateral $(x_i,c_i',c_i,m_i)$. Equation \eqref{eq:dmiplusp} follows from the hyperbolic Pythagorean theorem in triangle $(x_i^+,p,m_i^+)$. Moreover, applying the hyperbolic Pythagorean theorem in triangle $(x_i,c_i,m_i)$ yields
	\begin{align}
		\cosh(d(c_i,x_i))&=\cosh(\tfrac{1}{6}\ell(\gamma))\cosh(d(c,m_i)),\nonumber\\
		&=\cosh(\tfrac{1}{6}\ell(\gamma))\cosh\bigg(\arctanh\bigg(\dfrac{\cosh(\tfrac{1}{6}\ell(\gamma))}{\coth(\tfrac{1}{2}K_C)}\bigg)\bigg),\label{eq:dcixi}
	\end{align}
	where we used equation \eqref{eq:dmici} in the second line.
	
	When we substitute the expressions for $K_C$ and $d(x_i^+,x_{i+1}^+)$ into equations \eqref{eq:dmimiplus},\eqref{eq:dmici},\eqref{eq:dmiplusp} and \eqref{eq:dcixi}, we find expressions for $d(m_i,m_i^+),d(m_i,c_i),d(m_i^+,p)$ and $d(c_i,x_i)$ in terms of $\varepsilon$ and $\ell(\gamma)$. As $\varepsilon=0.72$ is fixed, we can treat these as functions of $\ell(\gamma)$. By a straightforward (but tedious) computation, it can be shown that $d(m_i,m_i^+)-d(m_i,c_i)-d(c_i,x_i)$ is strictly decreasing as a function of $\ell(\gamma)$ with minimum $-0.180\ldots$ for $\ell(\gamma)=2\varepsilon'$. By a similar computation, $d(m_i^+,p)$ is strictly increasing as a function of $\ell(\gamma)$ with minimum $0.247\ldots$ for $\ell(\gamma)\rightarrow 0$. We conclude that
	$$ d(m_i,m_i^+)-d(m_i,c_i)-d(c_i,x_i)+d(m_i^+,p)\geq -0.180\ldots+0.247\ldots>0,$$
	from which it follows that $d(c_i,p)=d(m_i,m_i^+)-d(m_i,c_i)+d(m_i^+,p)>d(c_i,p)$. Hence, $p$ is not contained in $D_i$. This finishes the proof.    	
\end{proof}

\section{Proof of Lemma~\ref{lem:edgebound}}\label{sec:appendixupperboundtight}

\begin{proof}
	\textbf{(Lemma~\ref{lem:edgebound})}\\
		Throughout the proof, we denote the set of vertices of $\mcT$ contained in a subset $U$ of $\mM$ by $V(U)$. Likewise, let $E(U,W)$ be the set of edges with one endpoint in $U\subset\mM$ and one endpoint in $W\subset\mM$.
		
		\bigskip\noindent 
		\textit{Part 1.} Consider the graph $G_i=(V(\mS_i),E(\mS_i,\mS_i))$. Let $g_i$ be the genus of $G_i$, i.e., the minimal genus of a surface onto which $G_i$ can be embedded. It is known that \cite[Proposition 4.4.4]{mohar2001}
		$$ g_i\geq \bigg\lceil\dfrac{e(\mS_i,\mS_i)}{6}-\dfrac{v(\mS_i)}{2}+1\bigg\rceil,$$
		or, equivalently, that 
		\begin{equation}\label{eq:graphedgeupperbound}
			e(\mS_i,\mS_i)\leq 6g_i+3v(\mS_i)-6.
		\end{equation} 
		We will show that the embedding of $G_i$ into $\mM$ intersects at most $6N(N+1)+2$ pairs of pants, which implies that $g_i\leq 3N(N+1)+1$. Certainly, $V(\mS_i)$ is contained in $\mS_i$, which consists of at most $6N$ consecutive pairs of pants. Now, let $e\in E(\mS_i,\mS_i)$. Because the diameter of $\mS_i$ is at most $6NM$, we know that there exists a path of length at most $6NM$ between the endpoints of $e$. Because $e$ is a distance path, it follows that 
		$$ \ell(e)\leq 6NM< 6N^2m.$$
		Suppose that $e$ intersects exactly $k$ pairs of pants that are not contained in $\mS_i$ and denote the farthest pair of pants that it intersects by $Y^*$. Here, farthest is defined with respect to the distance along the trivalent graph $L_g$. Because $e$ is an edge between vertices in $\mS_i$, $e$ has to traverse at least $k-1$ pairs of pants to reach $Y^*$ and similarly at least $k-1$ pairs of pants to return to $\mS_i$. As the length of $e$ within each of these pairs of pants is at least $m$, we know
		$$ \ell(e)\geq 2(k-1)m.$$
		It follows that 
		$$ 2(k-1)m < 6N^2m,$$
		which implies that $k<3N^2+1$. We conclude that $G_i$ is embedded in a surface consisting of at most $(3N^2+1)+6N+(3N^2+1)=6N(N+1)+2$ pairs of pants. It follows that $g_i\leq 3N(N+1)+1$. Hence 
		$$ e(\mS_i,\mS_i)\leq 6(3N(N+1)+1)+3v(\mS_i)-6=3v(\mS_i)+18N(N+1),$$
		which finishes the proof.
		
		\bigskip\noindent
		\textit{Part 2.} 
		We consider two cases. 
		
		\medskip\noindent
		\textit{Case 1: there are at most $6N^2+2$ pairs of pants between $\mS_i$ and $\mS_{i+1}$.}
		
		\noindent
		Consider the graph $(V(\mS_i\cup \mS_{i+1}),E(\mS_i\cup\mS_{i+1},\mS_i\cup\mS_{i+1}))$. We have shown in Part 1 that edges in $E(\mS_i,\mS_i)$ and $E(\mS_{i+1},\mS_{i+1})$ can traverse at most $3N^2-1$ pairs of pants that are not contained in $\mS_i$ and $\mS_{i+1}$ respectively. Therefore, $E(\mS_i,\mS_i)\cup E(\mS_{i+1},\mS_{i+1})$ is contained in a surface consisting of at most 
		$$(3N^2+1)+6N+(3N^2+1)+(3N^2+1)+6N+(3N^2+1)=12N^2+12N+4$$
		pairs of pants. With a similar argument it can be shown that $E(\mS_i,\mS_{i+1})$ is contained in this surface as well. Therefore, $(V(\mS_i\cup \mS_{i+1}),E(\mS_i\cup\mS_{i+1},\mS_i\cup\mS_{i+1}))$ is embedded in a surface consisting of at most $12N^2+12N+4$ pairs of pants. Replacing $\mS_i$ by $\mS_i\cup\mS_{i+1}$ in Inequality \eqref{eq:graphedgeupperbound} yields
		\begin{align*}
			e(\mS_i\cup\mS_{i+1},\mS_i\cup\mS_{i+1})&\leq 6(6N^2+6N+4)+3v(\mS_i\cup\mS_{i+1})-6,\\
			&=3v(\mS_i\cup\mS_{i+1})+36N(N+1)+18.
		\end{align*}
		Because $e(\mS_i,\mS_{i+1})\leq e(\mS_i\cup\mS_{i+1},\mS_i\cup\mS_{i+1})$, the desired inequality follows. 
		
		\bigskip\noindent
		\textit{Case 2: there are more than $6N^2+2$ pairs of pants between $\mS_i$ and $\mS_{i+1}$.} 
		
		\noindent
		We show that there are integers $g_{i,1}$ and $g_{i,2}$ with
		\begin{equation}
		\label{eq:gi1-leq-gi2}
		g_{i,1} \leq g_{i,2},
		\end{equation}
		such that
		\begin{equation}
		\label{eq:gi1}
		e(\mS_i,\mS_{i+1})\leq 2v(\mS_i)+2v(\mS_{i+1})+4g_{i,1}-4
		\end{equation}
		and
		\begin{equation}
		\label{eq:gi2}
		e(\mS_i,\mS_{i+1})\geq 3v(\mS_i)+3v(\mS_{i+1})-3e(\mS_i,\mS_i)-3e(\mS_{i+1},\mS_{i+1})+6g_{i,2}-6.
		\end{equation}
		Combining \eqref{eq:gi1} and \eqref{eq:gi2} yields
		\begin{equation*}
		e(\mS_i,\mS_{i+1})\leq 6 e(\mS_i,\mS_i) + 6 e(\mS_{i+1},\mS_{i+1})  +
		12g_{i,1} - 12g_{i,2}.
		\end{equation*}
		Using the upper bound $e(\mS_j,\mS_j) \leq 3 v(\mS_j) + 18 N(N+1)$ from Part 1 for $j = i$ and $j = i+1$, together with
		\eqref{eq:gi1-leq-gi2} yields
	    \begin{equation*}
	    e(\mS_i,\mS_{i+1})\leq 18 v(\mS_i \cup \mS_{i+1}) + 216 N(N + 1),
	    \end{equation*} 
	    which is the desired inequality.
	    
	   	The number $g_{i,1}$ is the genus of the graph $G_{i,1}$ with edge set $E(\mS_i,\mS_{i+1})$ and
	   	vertex set $V_{i,1}$ consisting of all vertices incident to some edge in $E(\mS_i,\mS_{i+1})$.
	   	The number $g_{i,2}$ is the genus of the graph $G_{i,2}$ with edge set
	   	$E(\mS_i\cup\mS_{i+1}, \mS_i\cup\mS_{i+1})$ and vertex set $V(\mS_{i+1}\cup\mS_{i+1})$.
	   	Since $G_{i,1}$ is a subgraph of $G_{i,2}$, inequality \eqref{eq:gi1-leq-gi2} holds.
	   	Therefore, it remains to be proved that \eqref{eq:gi1} holds for this value of $g_{i,1}$, and that
	   	\eqref{eq:gi2} holds for this value of $g_{i,2}$.
	   			
		To prove that \eqref{eq:gi1} holds, we apply a result in graph theory about bipartite graphs to $G_{i,1}$. By construction, $G_{i,1}$ contains no cycle of length 3. We claim that $G_{i,1}$ is connected. Then, \cite[Prop. 4.4.4, eq. 4.13]{mohar2001} 
		\begin{align*}
		e(\mS_i,\mS_{i+1})&\leq 2|V_{i,1}|+4g_{i,1}-4.
		\end{align*}
		Observing that $|V_{i,1}|\leq v(\mS_i)+v(\mS_{i+1})$ yields the desired inequality.
		
		To prove that $G_{i,1}$ is connected, consider a pair of pants $Y$ between $\mS_i$ and $\mS_{i+1}$ such that there are at least $3N^2+1$ pairs of pants between $Y$ and $\mS_i$ and between $Y$ and $\mS_{i+1}$. Such a pair of pants $Y$ exists because there are more than $6N^2+2$ pairs of pants between $\mS_i$ and $\mS_{i+1}$. Let $\gamma$ be the boundary geodesic of $Y$ such that $\mM\setminus\gamma$ consists of exactly two connected components. Every edge in $E(\mS_i,\mS_{i+1})$ intersects $\gamma$, because $\gamma$ separates $\mS_i$ and $\mS_{i+1}$. Furthermore, every edge in $E(\mS_i,\mS_{i+1})$ intersects $\gamma$ \emph{exactly once}, because the edges in $E(\mS_i,\mS_{i+1})$ and $\gamma$ are geodesics and there are no hyperbolic bigons. No edge in $E(\mS_i,\mS_i)\cup E(\mS_{i+1},\mS_{i+1})$ intersects $\gamma$, because by the reasoning in Part 1 edges in $E(\mS_j,\mS_j)$ for $j=i$ and $j=i+1$ intersect fewer than $3N^2+1$ pairs of pants that are not contained in $\mS_j$. Because $\mcT$ has no edges between non-adjacent clusters, it follows that the only edges of $\mcT$ that intersect $\gamma$ are edges in $E(\mS_i,\mS_{i+1})$ and that each edge intersects $\gamma$ in exactly one point. Therefore, we can write $E(\mS_i,\mS_{i+1})=\{e_1,\ldots,e_k\}$ for $k=|E(\mS_i,\mS_{i+1})|$, where the indices of the edges correspond to the order in which they intersect $\gamma$. For every $i=1,\ldots,k$, the edges $e_i$ and $e_{i+1}$ are contained in a triangle of $\mcT$ consisting of $e_i,e_{i+1}$ and an edge in either $E(\mS_i,\mS_i)$ or $E(\mS_{i+1},\mS_{i+1})$. Hence, $e_i$ and $e_{i+1}$ share an endpoint. Therefore, there is a path in $G_{i,1}$ between any endpoint of $e_i$ and any endpoint of $e_{i+1}$. This implies that there is a path in $G_{i,1}$ between any endpoint of $e_i$ and $e_j$ for all $i,j=1,\ldots,k$. Because $V_{i,1}$ consists of the vertices incident to some edge in $E(\mS_i,\mS_{i+1})$, it follows that $G_{i,1}$ is connected. This concludes the proof of Inequality \eqref{eq:gi1}.
		
		We continue with the proof of Inequality \eqref{eq:gi2}. Recall that the graph $G_{i,2}$ has edge set $E(\mS_i\cup\mS_{i+1}, \mS_i\cup\mS_{i+1})$ and vertex set $V(\mS_{i+1}\cup\mS_{i+1})$. The union of the triangles of $\mcT$ with edges and vertices in $G_{i,2}$ define a topological surface with boundary components\footnote{In fact, these triangles define a hyperbolic surface with boundary components consisting of a finite number of geodesic segments. However, for our argument we do not use any metric properties.}. Each boundary component consists of a finite number of edges of $G_{i,2}$. By adding a face to each boundary component, we obtain an embedding of $G_{i,2}$ in a closed surface $S_i$. When we speak of faces of $G_{i,2}$, we will always refer to faces with respect to the embedding of $G_{i,2}$ in $S_i$. 
		
		Denote the number of edges that are contained in exactly zero, one or two triangles in $G_{i,2}$ by $\delta_0,\delta_1$ and $\delta_2$, respectively. The total number of edges is given by $e_{i,2}=\delta_0+\delta_1+\delta_2$. Now, let $f_\Delta$ be the number of triangular faces of $G_{i,2}$. As the total number of faces $f_{i,2}$ is at least the number of triangular faces, we know that $f_{i,2}\geq f_\Delta$. Since $3f_\Delta=2\delta_2+\delta_1$, it follows that
		$$ 3f_{i,2}\geq 2\delta_2+\delta_1.$$
		As $e_{i,2}=\delta_0+\delta_1+\delta_2$, we obtain
		$$ 3f_{i,2}\geq 2(e_{i,2}-\delta_1-\delta_0)+\delta_1=2e_{i,2}-\delta_1-2\delta_0.$$
		Because
		\begin{equation}\label{eq:subgraphtotaledges}
		e_{i,2}=e(\mS_i,\mS_i)+e(\mS_i,\mS_{i+1})+e(\mS_{i+1},\mS_{i+1}),
		\end{equation}
		we see
		\begin{equation}\label{eq:subgraphtrianglelowerbound}
		3f_{i,2}\geq 2e(\mS_i,\mS_i)+2e(\mS_i,\mS_{i+1})+2e(\mS_{i+1},\mS_{i+1})-\delta_1-2\delta_0.
		\end{equation}
		We will now bound the right-hand side from below by proving that 	$$\delta_1+2\delta_0\leq 2e(\mS_i,\mS_i)+2e(\mS_{i+1},\mS_{i+1}).$$
		We claim that every edge in $E(\mS_{i},\mS_{i+1})$ is part of two triangles. To prove this, let $e=(u,v)\in E(\mS_i,\mS_{i+1})$ and consider a triangle $(u,v,w)$ in $\mcT$ containing $e$. Because every edge has its endpoints in either the same cluster or consecutive clusters, there is only an edge between $u$ and $w$ if $w\in\mS_j$ for $j=i-1,i,i+1$ and there is only an edge between $v$ and $w$ if $w\in\mS_j$ for $j=i,i+1,i+2$. Since $(u,v,w)$ is a triangle in $\mcT$, it follows that $w\in \mS_i\cup\mS_{i+1}$, so $w\in G_{i,2}$. This means that $(u,v,w)$ is a triangle in $G_{i,2}$ and since we have chosen it arbitrarily, both triangles in $\mcT$ containing $e$ are triangles in $G_{i,2}$. It follows that the edges that are contained in no triangles or exactly one triangle in $G_{i,2}$ are contained in $E(\mS_i,\mS_i)\cup E(\mS_{i+1},\mS_{i+1})$. This means that 
		$$\delta_0+\delta_1\leq e(\mS_i,\mS_i)+e(\mS_{i+1},\mS_{i+1}),$$
		so
		$$2\delta_0+\delta_1\leq 2e(\mS_i,\mS_i)+2e(\mS_{i+1},\mS_{i+1}).$$ 
		Combining this upper bound with Equation \eqref{eq:subgraphtrianglelowerbound} we obtain
		\begin{equation}\label{eq:subgraphfaceslowerbound}
		3f_{i,2}\geq 2e(\mS_i,\mS_{i+1}).
		\end{equation}
		To conclude, we will look at Euler's formula for the graph $G_{i,2}$, which is given by 
		\begin{equation}\label{eq:euler}
			v_{i,2}-e_{i,2}+f_{i,2}=2-2g'_{i,2},
		\end{equation} 
		where $g'_{i,2}$ is the genus of the embedding of $G_{i,2}$ in $S_i$. Because $g_{i,2}$ is the minimal genus of a surface onto which $G_{i,2}$ can be embedded, in particular $g_{i,2}\leq g'_{i,2}$. Substituting Equation \eqref{eq:subgraphtotaledges} and Inequality \eqref{eq:subgraphfaceslowerbound} into Euler's formula \eqref{eq:euler}, we obtain after some simplifications
		$$ e(\mS_i,\mS_{i+1})\geq 3v(\mS_i\cup\mS_{i+1})-3e(\mS_i,\mS_i)-3e(\mS_{i+1},\mS_{i+1})+6g_{i,2}-6.$$
		This finishes the proof.
\end{proof}

\bibliography{Hyperbolicsurfaces}	
\bibliographystyle{plainurl}

\end{document}